\documentclass[conference,twocolumn]{IEEEtran}%
\usepackage[T1]{fontenc}
\usepackage{graphicx}
\usepackage[cmex10]{amsmath}
\usepackage{amsthm}
\usepackage{amsfonts}
\usepackage{amssymb}
\usepackage{mathtools}
\usepackage{graphicx}
\usepackage{color}
\usepackage{algorithmic}
\usepackage{algorithm}
\usepackage[english]{babel}
\usepackage{blindtext}
\usepackage{epsfig}
\usepackage{framed}
\usepackage[table,usenames,dvipsnames]{xcolor}
\usepackage{tikz}
\usepackage{hyperref}%
\providecommand{\U}[1]{\protect\rule{.1in}{.1in}}
\providecommand{\U}[1]{\protect\rule{.1in}{.1in}}

\newtheorem{observation}{Observation}
\theoremstyle{plain}
\newtheorem{theorem}{Theorem}

\newtheorem{lemma}{Lemma}
\newtheorem{proposition}{Proposition}

\theoremstyle{definition}

\newtheorem{definition}{Definition}

\definecolor{ingo}{RGB}{139,0,0}

\definecolor{gerhard}{RGB}{0,0,139}

\definecolor{rick}{RGB}{0,139,0}

\renewcommand{\hbar}{h^\ast}
\newenvironment{figure2}{\begin{figure*}}{\end{figure*}}
\begin{document}

\title{Performance of Hierarchical Sparse Detectors for Massive MTC}
\author{Gerhard Wunder, Ingo Roth, Rick Fritschek and Jens Eisert\\{\footnotesize {Freie Universit\"at Berlin, Germany.} }\\[-0.3em] }
\maketitle

\begin{abstract}
Recently, a new class of so-called \emph{hierarchical thresholding algorithms}
was introduced to optimally exploit the sparsity structure in joint user
activity and channel detection problems. In this paper, we take a closer look
at the user detection performance of such algorithms under noise and relate
its performance to the classical block correlation detector with orthogonal
signatures. More specifically, we derive a lower bound for the diversity order
which, under suitable choice of the signatures, equals that of the block
correlation detector. Surprisingly, in specific parameter settings
non-orthogonal pilots, i.e. pilots where (cyclically) shifted versions
interfere with each other, outperform the block correlation detector.
Altogether, we show that, in wide parameter regimes, the hierarchical
thresholding detectors behave like the classical correlator with improved
detection performance but operate with much less required pilot subcarriers.
We provide mathematically rigorous and easy to handle formulas for numerical
evaluations and system design. Finally, we evaluate our findings with
numerical examples and show that, in a practical parameter setting, a
classical pilot channel can accommodate up to three advanced pilot channels
with the same performance.

\end{abstract}

%% ----------- metadata for hyperref -------------
\makeatletter\hypersetup{pdftitle = {HiHTP: A Custom-Tailored Hierarchical Sparse Detector for Massive MTC},
	     pdfauthor = {G. Wunder, I. Roth, R. Fritschek, J. Eisert},
	     pdfsubject = {Compressed sensing for 5G},
	     pdfkeywords = {Compressive sensing, model-based,
				     support recovery, hierarchically sparse,
				     block sparse, level sparse, hard thresholding pursuit,
				     iterative hard thresholding, sparse recovery, hierarchical hard thresholding pursuit, massive machine-type communication,
				     user identification, channel estimation
				     }
	    }\hypersetup{pageanchor=false}\makeatother

\section{Introduction}

Compressed Sensing (CS) is a mathematical field with applications in many
engineering disciplines involving \emph{big data processing}. One of the
recent intriguing fields of applications is 5G \& Beyond wireless
communication, particularly the so-called massive Machine-type Communications
(mMTC) scenario \cite{Wunder2015_ACCESS}. In mMTC, CS is used as an advanced
nonlinear multiuser detector (in uplink) which takes advantage of the sparse
user activity as well as the sparse channel profiles. Thereby, it can resolve
overload situations and identify the active user set 'en passant', in clear
contrast to the classical detectors.

Meanwhile, there is a large body of literature on such detection problems,
often termed \emph{compressive random access} \cite{Wunder2015_ACCESS}.
Initial references are \cite{Applebaum2012_PHYCOM,Zhu2011_TWC}, followed up by
major work of Bockelmann et al. \cite{Bockelm2013_ETT,Yi2014_GC} and recently
by Choi \cite{Choi2017_IoT,Choi2017_TWC}. A \emph{one (or single) shot}
approach has been proposed in
\cite{Wunder2014_ICC,Wunder2015_GC,Wunder2015_ASILOMAR,Wunder2017_ASILOMAR}
where both data and pilot channels are overloaded within the same OFDM symbol.
A comprehensive overview of competitive approaches within 5G (Rel. 16 upwards)
can be found in \cite{Bockelm2018_ACCESS,Bockelm2016_COMMAG,Schaich2016_ETT}.

While each of the approaches take specific properties into account, there is a
fairly general signal model. Most common is the block structure (users) and
within-block structure (channel taps) but surprisingly a custom-tailored
provably converging detector has not been known until very recently with the
invention of the \emph{Hierarchical Hard Thresholding Pursuit} (HiHTP)
algorithm in Ref.~\cite{RothEtAl:2016}. The key observation of HiHTP is that
sporadic user activity and sparse channel profiles give rise to a
hierarchically sparse structured vector consisting of all estimated channel
coefficients. Motivated by this observation, HiHTP can efficiently reconstruct
hierarchically sparse signals from only a small number of linear measurements.
In Ref.~\cite{RothEtAl:2016} recovery guarantees are derived for HiHTP
(Theorem \ref{thm:HiHTP}) and its performance is compared to other algorithms,
e.g.~the HiLasso-Algorithm \cite{SprechmannEtAl:2011}. Recently, we were also
informed that in \cite{Schepker2013_ISWCS} (and follow up work) a hierarchical
version of the \emph{Orthogonal Matching Pursuit} (OMP) algorithm was invented
in the same context but without providing a proof of convergence. Notably, an
even simpler variant of HiHTP is the related \emph{Hierarchical Iterative Hard
Thresholding} (HiIHT) algorithm. This was introduced and analyzed in our
recent paper \cite{Wunder2018_ARXIV-2} which studies an application featuring
multiple levels in the hierarchy arising from considering multiple antennas
and multiple measurements.

In this paper, we take a closer look at the user detection performance of such
hierarchical detectors under noise. We relate its performance to the classical
block correlation detector with orthogonal signatures \cite{Lee2013_TransWC}.
More specifically:

\begin{itemize}
\item We derive a lower bound for the diversity order of HiHTP/HiIHT
algorithms which, under suitable choice of the signatures, equals that of the
block correlation detector (Theorem \ref{thm:md1}). Naturally, a notion of
diversity order only makes sense if there are sufficiently many compressive
measurements so that the `outage probability' of HiHTP/HiIHT becomes
negligible in noise. Our first result is that non-orthogonal pilots, i.e.
pilots with mutually interfering shifted versions, outperform orthogonal
pilots in regimes characterized by a large user set/delay spread and a
sublinear growth of user activity/diversity, respectively (Theorem
\ref{thm:asymp}). To underline this fact: this holds true even when \emph{all}
received samples were available for the detection as it is for the block
correlation detector. This is a clear discrepancy to the non-sparse situation
where the block correlation detector is optimal.

\item Motivated by numerical evidence we show that user detection performance
of HiHTP/HiIHT is actually independent of user activity for a wide range of
parameter settings (but strongly depends on the channel profile of course). We
carry out an extended analysis heavily relying on concentration of measure
inequalities and prove that the number of sufficient compressive measurements
is at worst only slightly penalized for this to hold (Theorem \ref{thm:md2}).
Consequently, HiHTP/HiIHT essentially behaves like the classical correlator
with improved detection performance and with much less required pilot
subcarriers. The bottomline here is that user capacity is drastically
increased since the remaining subcarriers can implement further pilot
channels. In the simulation section, we validate this for practical settings.
\end{itemize}

The remainder of the paper is structured as follows: After introducing the
system model in Sec. \ref{sec:system} and the algorithms in Sec.
\ref{sec:algorithms}, we make our statements mathematically rigorous and
provide explicit and easy to handle formulas for numerical evaluations and
system design in Sec. \ref{sec:performance}. Finally, we verify our findings
with simulations and evaluations and conclude.

\textbf{Notations}. Let $\lVert x\rVert_{\ell_{q}}=(\sum_{i}|x_{i}|^{q}%
)^{1/q}$, $q>0$, be the $\ell_{q}$-norms and $\lVert x\rVert:=\lVert
x\rVert_{\ell_{2}}$. We use the short hand notation $[N]$ for the set
$\left\{  0,1,\ldots,N-1\right\}  $ and denote for any set $\mathcal{A}$ its
cardinality by $|\mathcal{A}|$. The elements of a vector/sequence $x$ are
referred to as $(x)_{i}$ or simply $x_{i}$ if it clear from the context. The
vector $x_{\mathcal{A}}\in\mathbb{C}^{|\mathcal{A}|}$ (matrix $X_{\mathcal{A}%
}$) is the projection of elements (rows) of the vector $x\in\mathbb{C}^{n}$
(matrix $X$) onto the subspace indexed by $\mathcal{A}\subset\lbrack n]$.
Depending on the context we also denote by $x_{\mathcal{A}}$ the vector that
coincides with $x$ for the elements indexed by $\mathcal{A}$ and is zero
otherwise. $I_{n}$ is the $n\times n$ identity matrix and $\operatorname{diag}%
(x)$ is the diagonal matrix with the vector $x\in\mathbb{C}^{n}$ on its
diagonal. For a matrix $A$, $A^{H/T}$ is its adjoint/transpose. The
multivariate complex Gaussian distribution of zero mean and covariance matrix
$\sigma^{2}I_{n}$ is denoted by $\mathcal{CN}\left(  0,\sigma^{2}I_{n}\right)
$. A vector $x\in\mathbb{C}^{N}$ is called $s$-sparse if it consists of at
most $s$ non-zero elements. The set of non-zero elements (support) of
$x\in\mathbb{C}^{N}$ is denoted as $\text{supp}(x)$. The imaginary unit is
$\imath=\sqrt{-1}$.

\section{System Model}

\label{sec:system}

Joint detection problems of mMTC, say in 5G uplink, can typically be casted as
follows: We allow for a fixed maximum set of $u$ users in a system with a
signal space of total dimension $n$, which can possibly be very large, e.g.
$2^{14}$ \cite{Wunder2015_GC}. The (time domain) signature $p_{i}\in
\mathbb{C}^{n}$ of the $i$-th user is taken from a possibly large set
$\mathcal{P}\subset\mathbb{C}^{n}$. Let $h_{i}\in\mathbb{C}^{s}$ denote the
sampled channel impulse response (CIR) of the $i$-th user, where $s\ll n$ is
the length of the cyclic prefix. While active users have a non-vanishing CIR,
inactive users are modeled by $h_{i}=0$. Furthermore, we define the matrix
$\text{circ}^{(s)}(p_{i})\in\mathbb{C}^{n\times s}$ to be the circulant matrix
with $p_{i}$ in its first column and $s-1$ shifted versions in the remaining
columns. Stacking the CIRs into a single column vector $h=[h_{1}^{T}%
\ h_{2}^{T}\ldots h_{u}^{T}]^{T}$, the signal received by the base station is
given by%
\[
y=D(p)h+e,
\]
where
$D(p)\coloneqq[\operatorname{circ}^{(s)}(p_{1}),\dots ,\operatorname{circ}^{(s)}(p_{u})]\in
\mathbb{C}^{n\times us}$ depends on the stacked signatures
$p\coloneqq[ p_{1}^{T}\ p_{2}^{T}\ldots p_{u}^{T}]^{T}$. In addition,
$e\in\mathbb{C}^{n}$ is assumed to be additive white Gaussian noise, i.e.
$e\sim\mathcal{CN}(0,\sigma^{2}I_{n})$.

A key idea in compressive random access is that the user identification and
channel estimation task needs to be accomplished within a much smaller
subspace, compared to the signal space, so that the remaining dimensions can
be exploited. The measurements in this subspace are of the form:
\begin{equation}
y_{\mathcal{B}}=\Phi_{\mathcal{B}}y=\Phi_{\mathcal{B}}D(p)h+\Phi_{\mathcal{B}%
}e, \label{eq:SysMod}%
\end{equation}
where we denote the restriction of some measurement matrix to a set of rows
with indices in $\mathcal{B}\subset\lbrack n]$ by $\Phi_{\mathcal{B}}$. In
practice, randomized (normalized) FFT measurements, $\Phi_{\mathcal{B}%
}=W_{\mathcal{B}},(W)_{ij}:=n^{-\frac{1}{2}}e^{-\imath2\pi ij/n}$ for
$k,l=0\dots n-1$, are typically implemented.

All performance indicators of the scheme strongly depend on the size of the
control window $\mathcal{B}$ and its complement $\mathcal{B}^{C}$ where
$\mathcal{B}\cup\mathcal{B}^{C}=[n]$. It is desired to keep the size of the
observation window $m\leq|\mathcal{B}|$ as small as possible to reduce the
control overhead $m/n$. The unused subcarriers in $\mathcal{B}^{C}$ can then
be used to implement further \emph{parallel} control channels for, say, user
activity detection. We call the ratio $n/m$\% the \emph{user capacity gain}.
In other words if for the same detection performance only $m<n$ subcarriers
are required then $n/m$\% more users can be detected.

The task of user identification amounts to the inverse problem of estimating
the non-vanishing blocks of $h$. The number of subcarriers in $\mathcal{B}$
required for solving this inverse problem depends on the structure of the
measurement map $\Phi_{\mathcal{B}}D(p)$ and the structure of $h$. In the
remainder of this section, we discuss properties of the measurement map in an
important special case and the sparsity structure of $h$.

\subsection{Proxy measurement model}

In general, the measurement map is difficult to analyze since $D(p)$ in eq.
(\ref{eq:SysMod}) depends on the specific design of the signatures $p_{i}$.
Assuming that $n\geq us$, we can define the signature set $\mathcal{P}$ in the
following way: We choose $p_{0}$ to be a sequence with unit power in frequency
domain such that
\begin{equation}
|\left(  \hat{p}_{0}\right)  _{i}|=\left\{
\begin{array}
[c]{cc}%
\sqrt{\frac{n}{m}} & i\in\mathcal{B}\\
0 & \text{else}%
\end{array}
\right.  , \label{eq:signa}%
\end{equation}
where $\hat{p}_{0}:=Wp_{0}$ denotes the FFT transform of $p_{0}$. Since $n\geq
us$, the matrix $D(p)$ can be completely composed of cyclical shifts of the
sequence $p_{0}$, i.e.:%
\[
p_{1}=p_{0},\quad p_{2}=p_{1}^{\left(  s\right)  },\quad p_{3}=p_{2}^{\left(
s\right)  },\quad\ldots\ ,
\]
where $p^{\left(  i\right)  }$ is the $i$ times cyclically shifted $p$. Hence,
$D(p)$ is a single circulant matrix. Clearly, $D(p)=W^{H}$diag$\left(
\sqrt{n}\hat{p}_{0}\right)  W$ and we can write%
\begin{align}
y_{\mathcal{B}}  &  =\Phi_{\mathcal{B}}D(p)h+\Phi_{\mathcal{B}}e\nonumber\\
&  =\Phi_{\mathcal{B}}W^{H}\text{diag}\left(  \sqrt{n}\hat{p}_{0}\right)
Wh+\Phi_{\mathcal{B}}e\nonumber\\
&  =\text{diag}\left(  \sqrt{n}\hat{p}_{\mathcal{B}}\right)  Wh+\Phi
_{\mathcal{B}}e\nonumber\\
&  \Longleftrightarrow y_{\mathcal{B}}^{\prime}=Ah+z, \label{eq:1stModel}%
\end{align}
where $A$ can be regarded as a randomized subsampled FFT\footnote{In fact, the
rows of $A$ are decorated with an additional factor given by the phases of
$\left(  \hat{p}_{0}\right)  _{i}$. However, these are not important for the
remaining analysis.}, which is normalized by an additional factor of
$\sqrt{n/m}$. Under the assumption that the additive noise $e$ is Gaussian
with variance $\sigma^{2}$, we find that $z\sim\mathcal{CN}\left(
0,\frac{\sigma^{2}}{n}I_{m}\right)  $.

\begin{observation}
We emphasize here that by such choice of signatures the sequences are no
longer (circular) shift-orthogonal. This situation is different from the LTE
standard, where Frank-Zadoff-Chu shift-orthogonal sequences are used. However,
we will see, that due to the structure of $h$ this choice does not induce
performance losses and is even better in certain parameter regimes compared to
the shift-orthogonal case.
\end{observation}

We shall use a second model which turns out to be quite convenient in the
extended analysis. Here, we set $us=n$ possibly by appending zeros to vector
$h$. Clearly, we are loosing some denoising performance which is negligible
though. Eventually, since $A$ is a sub-sampled FFT matrix, it is clear that we
can write the system also in the following form:%
\begin{align}
y_{\mathcal{B}}  &  =Ah+z\nonumber\\
&  =A\left(  h+z^{\prime}\right) \nonumber\\
&  \Longleftrightarrow y_{\mathcal{B}}^{\prime\prime}=A\left(  h+z^{\prime
}\right)  , \label{eq:2edModel}%
\end{align}
where now $z^{\prime}\sim\mathcal{CN}\left(  0,\frac{\sigma^{2}m}{n^{2}}%
I_{n}\right)  $ by the properties of $A$, i.e. the noise per signal dimension
and its expectation is highly damped.

\subsection{Sparse priors}

The possibility to reconstruct $h$ from only a small control window
$\mathcal{B}$ relies on two structural assumptions \cite{Wunder2015_ACCESS}:
In mMTC we expect to have a large number of users having only sporadic
traffic. In effect, at a given time, only a small number of users $k_{u}\ll u$
is active. Therefore $h$ has only $k_{u}$ non-vanishing blocks. At the same
time the CIRs are observed to be sparse indicating that the blocks of $h$ have
at most $k_{s}$ non-vanishing entries. This leads to the following model for
$h$:

\begin{itemize}
\item The \emph{non-zero} complex-valued channel coefficients are independent
{normal distributed} with power $\mathbb{E}|(h_{i})_{j}|^{2}=\sigma_{h}^{2}$.

\item The support of $h_{i}$ is bounded with high probability such that we can
assume that $h_{i}\in\mathbb{C}^{s}$ with $s\ll n$ (typically, we have
$s\in\{300,\ldots,3000\}$).

\item The blocks $h_{i}$ are sparse, i.e. $|\operatorname{supp}(h_{i})|\leq
k_{s}$ (typically, we have $k_{s}\leq6$). Of course, practically this means
that only most of the energy is concentrated in $k_{s}$ paths, typically 95\%.
This is in accordance with current channel profiles, see e.g.
\cite{Ling2016_ARXIV} for a discussion of this assumption. There it is
reported that e.g. for a (rural) 6MHz bandwidth channel the assumptions indeed
hold for not to large delay spread, say below $6\mu s$. Furthermore, we point
out that in particular in our second model, in eq. (\ref{eq:2edModel}), the
remaining energy can be simply subsumed in the noise.

\item The support of the channels is uniformly distributed within $n$, {i.e.}
any subset has probability $\binom{s}{k_{s}}^{-1}$.

\item The user activity is sparse, i.e. $k_{u}$ users out of $u$ are actually
active (typically, we have $k_{u}\leq10$ out of $100$).

\item The set of active users is uniformly distributed within $n$, i.e.\ any
combination of users has probability $\binom{u}{k_{u}}^{-1}$ to be active.
\end{itemize}

The vector $h$ containing all CIRs has at most $k_{u}\cdot k_{s}$ non-vanishing
entries in total. But the hierarchical structure of the non-vanishing entries
of $h$ is even more restrictive. We give the following formal definition:

\begin{definition}
[Hierarchical sparsity] A compound vector $h \in\mathbb{C}^{u\cdot s}$
consisting of $u$ blocks of size $s$ is \emph{hierarchically $(k_{u},k_{s}%
)$-sparse} if at most $k_{u}$ blocks have non-vanishing entries and each of
these blocks is $k_{s}$-sparse.
\end{definition}

For convenience, we will call a hierarchically $(k_{u},k_{s})$-sparse vector
simply $(k_{u},k_{s})$\emph{-sparse}. Our signal model, thus, implies that $h$
is $(k_{u},k_{s})$-sparse. The support set of a $(k_{u},k_{s})$-sparse vector
is also called $(k_{u},k_{s})$-sparse. Let the actual set of active users and
active paths be $\mathcal{A}=\operatorname{supp}(h)$. We shall denote the set
of the active user indices by $\mathcal{A}^{B}$ and the non-vanishing path
locations of the $i$-th user by $\mathcal{A}_{i}$.

Note that we assume from now on, if not otherwise explicitly stated, that the
channel energy is equally distributed within the coefficients, which however
does not affect the generality of the results. In fact, all the results can be
formulated within the general framework, only the expressions become more
complicated. Hence, we shall set without loss of generality $\sigma_{h}^{2}=1$
so that the Signal-to-Noise Ratio ($\operatorname{SNR}$) becomes%
\[
\operatorname{SNR}:=\mathbb{E}|h_{0}|^{2}/\sigma^{2}=1/\sigma^{2}.
\]
Note, that $\operatorname{SNR}$ does not reflect the true receive
$\operatorname{SNR}$ in the system, which is $k_{s}/\sigma^{2}$.

In the next section, we discuss related algorithms HiHTP and HiIHT which take
advantage of the hierarchical sparsity.

\section{Detection algorithms}

\label{sec:algorithms}

\subsection{Block correlator}

Before we derive results for the new detection scheme, let us recapture the
approach of a simple block correlation detector. To this end, we define the
thresholding operator $L_{\xi}^{B}$. To a given compound vector $h\in
\mathbb{C}^{us}$, the operator $L_{\xi}^{B}$ assigns the subset $L_{\xi}%
^{B}(h)\subset\lbrack u]$ of indices of the blocks that exceed a threshold
$\sqrt{\xi}$ in $\ell_{2}$-norm, i.e.%
\begin{equation}
i\in L_{\xi}^{B}(h)\Longleftrightarrow\lVert h_{\mathcal{A}_{i}}\rVert
\geq\sqrt{\xi}. \label{eq:blocks}%
\end{equation}
Now, the $i$-th user is detected by the block correlation detector that
received a vector $y$ if
\begin{equation}
\sum_{j=0}^{s-1}|\langle y,p_{i}^{\left(  j\right)  }\rangle|^{2}\geq
\xi.\label{eq:corr}%
\end{equation}
Hence, the detection scheme is equivalent to defining the set of identified
active users as $L_{\xi}^{B}(D(p)^{H}y)$. In other words, the detector chooses
the users with energy collected over all shifts within the delay spread
exceeding a threshold. This method crucially relies on the orthogonality of
signatures excluding cross-talk from other signatures.

\subsection{HiHTP/HiIHT algorithm}

Motivated by the application in mMTC, the recovery of $(k_{u},k_{s})$-sparse
signals from linear measurements was studied in Ref.~\cite{RothEtAl:2016}
following the outline of model-based compressed
sensing~\cite{BaraniukEtAl:2010}. Therein an efficient algorithm, HiHTP, was
proposed and a recovery guarantee based on generalised restricted isometry
property (RIP) constants was proven. The main ingredient to the algorithm is
the $\ell_{2}$-norm projection onto hierarchical sparse vectors. For a vector
$x$ we denote the thresholding operator that gives the support of the best
$(k_{u},k_{s})$-sparse approximation to $x$ by
\begin{equation}
L_{k_{u},k_{s}}(x):=\operatorname{supp}\underset{\text{$(k_{u},k_{s})$-sparse
$y$}}{\operatorname*{arg}\min}\Vert{x-y}\Vert.\label{eq:threshold_1}%
\end{equation}
This operator can be efficiently calculated by selecting the $k_{s}$
absolutely largest entries in each block and subsequently the $k_{u}$ blocks
that are largest in $\ell_{2}$-norm. The strategy of the HiHTP algorithm is to
use the thresholding operator $L_{k_{u},k_{s}}$ to iteratively estimate the
support of $h$ and subsequently solve the inverse problem restricted to the
support estimate, see Algorithm~\ref{alg:HiHTP}.

\begin{algorithm}
[tb] \caption{HiHTP with user detection} \label{alg:HiHTP}

\begin{algorithmic}
[1]

\REQUIRE measurement matrix $A$, measurement vector $y_{\mathcal{B}}$,
sparsity $\left(  k_{u},k_{s}\right)  $, energy threshold $\xi$

\STATE$h^{\left(  t\right)  }=0$

\REPEAT

\STATE$\mathcal{A}^{\left(  t+1\right)  }=L_{k_{u},k_{s}}\left(  h^{\left(
t\right)  }+A^{H}\left(  y_{\mathcal{B}}-Ah^{\left(  t\right)  }\right)
\right)  $

\STATE$h^{\left(  t+1\right)  }=\arg\min{}_{z\in\mathbb{C}^{n}%
,\operatorname{supp}(z)\subseteq\mathcal{A}^{\left(  t+1\right)  }}\left\{
\Vert y_{\mathcal{B}}-Az\Vert\right\}  $

\STATE$t:=t+1$

\UNTIL stopping criterion is met at $t=t^{\ast}$

\STATE$\hbar=h^{\left(  t^{\ast}\right)  }$

\STATE$\mathcal{A}^{\ast}_{B}=L_{\xi}^{B}(\hbar)$

\ENSURE$(k_{u},k_{s})$-sparse vector $\hbar$ and active user set
$\mathcal{A}^{\ast}_{B}$.
\end{algorithmic}
\end{algorithm}

The HiHTP has a compagnion algorithm called HiIHT which is an even simpler
variant and given as Algorithm~\ref{alg:HiIHT}. The main difference is the
gradient step which involves a \emph{least squares minimization} step for
HiHTP but is omitted for HiIHT. We note that the performance analysis carries over verbatim to the HiIHT algorithm, see \cite{Wunder2018_ARXIV-2}.

\begin{algorithm}
[tb] \caption{HiIHT with user detection} \label{alg:HiIHT}

\begin{algorithmic}
[1]

\REQUIRE measurement matrix $A$, measurement vector $y_{\mathcal{B}}$,
sparsity $\left(  k_{u},k_{s}\right)  $, energy threshold $\xi$

\STATE$h^{\left(  t\right)  }=0$

\REPEAT

\STATE$\mathcal{A}^{\left(  t+1\right)  }=L_{k_{u},k_{s}}\left(  h^{\left(
t\right)  }+A^{H}\left(  y_{\mathcal{B}}-Ah^{\left(  t\right)  }\right)
\right)  $

\STATE$h^{\left(  t+1\right)  }=\left[ h^{\left(  t\right)  }+A^{H}\left(
y_{\mathcal{B}}-Ah^{\left(  t\right)  }\right) \right] _{\mathcal{A}^{\left(
t+1\right)  }} $

\STATE$t:=t+1$

\UNTIL stopping criterion is met at $t=t^{\ast}$

\STATE$\hbar=h^{\left(  t^{\ast}\right)  }$

\STATE$\mathcal{A}^{\ast}_{B}=L_{\xi}^{B}(\hbar)$

\ENSURE$(k_{u},k_{s})$-sparse vector $\hbar$ and active user set
$\mathcal{A}^{\ast}_{B}$.
\end{algorithmic}
\end{algorithm}

HiHTP/HiIHT algorithms come with a guarantee for stable and robust recovery
provided that the measurement matrix $A$ has the so-called hierarchical RIP
property custom tailored to the set of $(k_{u},k_{s})$-sparse vectors (for
details, please see \cite{RothEtAl:2016}). To date hierarchical RIP was not
shown for FFT measurements so that only the standard results apply. Needless
to say, there is also no RIP analysis for the more general situation where the
measurement matrix $\Phi_{\mathcal{B}}D(P)$ has a more complicated dependency
on the signature design. A hierarchical RIP bound for a measurement matrix
with {i.i.d.} Gaussian entries was derived in Ref.~\cite[Theorem
1]{RothEtAl:2016}. The result can be stated as the following theorem.

\begin{theorem}
\label{thm:HiHTP} Given an $(k_{u},k_{s})$-sparse vector $h\in\mathbb{C}^{us}$
and measurements of the form $y_{\mathcal{B}}=Ah+e$, where $A$ is a $m\times
us$ matrix with real {i.i.d.} Gaussian entries, the output $\hbar$ of
HiHTP/HiIHT, Algorithm~\ref{alg:HiHTP}, fullfils:
\[
\mathbb{P}(\Vert\hbar-h\Vert>\epsilon)\leq\mathbb{P}_{\overline{\text{RIP}}%
}+\mathbb{P}(\tau\Vert e\Vert>\epsilon).
\]
The probability that $A$ does not have the required hierarchical RIP property
$\mathbb{P}_{\overline{\text{RIP}}}$ is bounded by
\[
\mathbb{P}_{\overline{\text{RIP}}}\leq C\left(  \frac{es}{k_{s}}\right)
^{k_{s}}\left(  \frac{eu}{k_{u}}\right)  ^{k_{u}}e^{-cm}%
\]
with $C$ and $c$ independent numerical finite constants and $\tau
=\tau(m)<\infty$.
\end{theorem}

The parameter $\tau\left(  m\right)  $ is a noise enhancement which depends
crucially on the number of measurements. Typically, theoretical estimates for
$\tau\left(  m\right)  $ are too coarse compared to the actual performance in
numerical tests. Theorem~\ref{thm:HiHTP} can be equivalently stated as the
requirement
\[
m\gtrsim k_{u}\log({u}/{k}_{u})+k_{u}k_{s}\log({s}/k_{s})
\]
on the asymptotic scaling of the number of samples $m$ to guarantee recovery
of $h$ up to noise. Hence, the vector of all CIRs is correctly reconstructed
up to a noise induced error provided that the control window $\mathcal{B}$ has
a sufficient size of $m$. Note that HiHTP/HiIHT is concerned with the
reconstruction of all CIRs. Obviously, from the reconstructed CIRs the set of
active user can be determined in a final second step. We define the set of
active users identified by HiHTP/HiIHT as $\mathcal{A}^{\ast}_{B}:=L_{\xi}%
^{B}(\hbar)$, where $\hbar$ is the output and $L_{\xi}^{B}$ is defined in eq.
(\ref{eq:blocks}).

We now turn to the main part of the paper which contains a discussion of
relevant metrics and the respective performance analysis.

\section{Performance analysis}

\label{sec:performance}

\subsection{Figures of merit}

Theorem~\ref{thm:HiHTP} provides a full characterization of HiHTP/HiIHT for
the joint recovery of the channels of \emph{all} users. Intuitively, the
benefit of the algorithm becomes obvious if the performance per block is
considered. To this end, we denote the probability that a user is not
correctly detected as active or inactive by $P_{\text{be}}(\xi)$. Since all
blocks of $h$ are statistically equivalent and a user is active with
probability $k_{u}/u$, one concludes that for some user with index $i$
\[
\mathbb{P}\left(  \Vert\hbar_{i}-h_{i}\Vert>\epsilon\right)  \leq
P_{\text{be}}(\xi)+\frac{k_{u}}{u}\left(  \mathbb{P}_{\overline{\text{RIP}}%
}+\mathbb{P}(\tau\Vert e\Vert>\epsilon)\right)  ,
\]
where correspondingly $\hbar_{i}$ denotes the $i$-th block of $\hbar$. The
bound suggest that the performance is dominated by the second term for
realistic $\operatorname{SNR}$, which yields a $k_{u}/u$ gain over Theorem 1
in the task of channel estimation. However, it is dominated by the first term
$P_{\text{be}}(\xi)$ for large but realistic SNR, which describes the
interplay of the noise and the channel energy, assuming that the error floor
induced $\mathbb{P}_{\overline{\text{RIP}}}$ is negligible.

In order to bound $P_{\text{be}}(\xi)$, we define the probability that a
certain active user is missed by the user identification scheme by
$P_{\text{md}}(\xi)$. Note that by symmetry the probability of a missed
detection is identical for all active users and depends on the energy
threshold $\xi$. Correspondingly, the overall probability that any active user
is misdetected is bounded by $k_{u}P_{\text{md}}(\xi)$. But due to the
complicated interdependencies, this bound is not tight. The events are neither
mutually exclusive nor do they contain each other. The missed detection
probability per user $P_{\text{md}}\left(  \xi\right)  $ is a key metric for
the system \cite{Lee2013_TransWC}. Eventually, we define the probability that
(overall) some inactive user is falsely detected as active by $P_{\text{fa}%
}\left(  \xi\right)  $ so that finally
\[
P_{\text{be}}(\xi)\leq k_{u}P_{\text{md}}(\xi)+P_{\text{fa}}\left(  \xi\right)
.
\]
As in \cite{Lee2013_TransWC}, our analysis concentrates on $P_{\text{md}}$ in
the following since our tools can be easily applied to bound $P_{\text{fa}%
}\left(  \xi\right)  $ as well. To this end note that we have the upper bound
$P_{\text{fa}}\left(  \xi\right)  \leq\mathbb{P}\left(  \tau^{2}\Vert
e\Vert^{2}>\xi\right)  $. Using $e\sim\mathcal{CN}\left(  0,\sigma^{2}%
I_{m}\right)  $ and the concentration inequalities
\eqref{eq:Mconc1}\,--\,\eqref{eq:Mconc4} in the appendix the false alarm
probability of HiHTP/HiIHT can be bounded from above as%
\begin{align*}
P_{\text{fa}}\left(  \xi\right)   &  \leq e^{-\left(  \frac{n\xi}{\tau
^{2}m\sigma^{2}}-1\right)  ^{2}\frac{m}{2}}\\
&  =e^{-\left(  \frac{\operatorname{SNR}n\xi}{\tau^{2}m}-1\right)  ^{2}%
\frac{m}{2}}.
\end{align*}
The bound is not tight but still sufficent to adjust the threshold $\xi$. Once
we have correctly detected the active users, we can evaluate for each active
user $i$ the unnormalized frequency domain Mean Squared Error (MSE):
\begin{equation}
\text{MSE}_{i}:=\mathbb{E}\left\Vert \sqrt{n}W\left(  \hbar_{i}-h_{i}\right)
\right\Vert ^{2} .\label{eq:mse}%
\end{equation}
From the Theorem~\ref{thm:HiHTP} we have the upper bound $\text{MSE}_{i}%
\leq\tau^{2}m\sigma^{2}/n$. Moreover, eventually, we can invoke Theorem 2 in
\cite{Wunder2015_GC} to get an estimate of the achievable average subcarrier
rate $R_{i}$ (i.e. for those subcarriers in $\mathcal{B}^{C}$) based on the
MSE bound as%
\begin{align*}
R_{i}  &  \geq\mathbb{E}\left[  \log\left(  1+k_{s}\operatorname{SNR}\right)
\right]  -\log\left(  1+\text{MSE}_{i}\right) \\
&  \geq\mathbb{E}\left[  \log\left(  1+k_{s}\operatorname{SNR}\right)
\right]  -\log\left(  1+\tau^{2}m\sigma^{2}/n\right)  ,
\end{align*}
provided the user is detected (which happens with probability $P_{\text{md}}%
$). Since all terms are known except $P_{\text{md}}$, we shall now concentrate
on $P_{\text{md}}$.

\subsection{Baseline: The classical correlation detector}

In the end, we want to compare our final result with an exemplary result of
the recent literature \cite{Lee2013_TransWC}. In~\cite{Lee2013_TransWC} an
algorithm was presented which exploits the constant amplitude zero auto
correlation property of Zadoff-Chu sequences for signature identification. The
algorithm finds the maximal cross-correlation between the received signals and
the shifted sequences. Moreover, an exact analysis of the probability of
identification failure was derived for $\xi=0$ and the high SNR regime
\begin{align}
&  \log P_{\text{md}}\left(  \xi\right)  |_{\xi=0}\IEEEnonumber\\
&  \leq k_{s}\log\left(  \frac{1}{n\sigma^{2}}\right)  +\sum_{i=0}^{k_{s}}%
\log\left(  n\left(  n\sigma_{h}^{2}+\sigma^{2}\right)  \right)  +\log
B_{0}\IEEEnonumber\\
&  =-k_{s}\log\left(  1+n\operatorname{SNR}\right)  +\log B_{0}\left(
s,u,k_{s}\right)  ,\label{Zadoff-Chu}%
\end{align}
where $B_{0}$ is a constant that does only depend on the parameters $s$ and
$u$ and is given by \begin{IEEEeqnarray*}{rCl} \label{eq:b0}
B_{0}\left(  s,u,k_s\right)  &=&\frac{1}{\Gamma\left(  k_s\right)  }\sum_{i=1}
^{u-1}\left(  -1\right)  ^{i+1}\binom{u-1}{i}\nonumber\\
&\times&\>\sum_{j_{1}=1}^{s-1}
\cdots\sum_{j_{i}=1}^{s-1}\frac{\Gamma\left(  \sum_{k=1}^{i}j_{k}+k_s\right)
}{\prod_{k=1}^{i}\Gamma\left(  j_{k}+1\right)  }s^{-\sum_{k=1}^{i}j_{k}-k_s}.
\end{IEEEeqnarray*}We use this result as a baseline for comparison. We call
the pre-log factor the \emph{diversity order} of the detection scheme. Note
that the term $B_{0}(s,u,k_{s})$ is quite difficult to evaluate. We provide
simpler expressions in the following section.

\subsection{Missed detection rate of HiHTP}

Our bound for the missed detection probability of HiHTP/HiIHT is summarized in
the following theorem. Here, $F(\xi):=\mathbb{P}(\Vert h_{i}\Vert^{2}\leq\xi)$
is the cumulative distribution function of the norm of each of the blocks
(which is independent of $i$).

\begin{theorem}
\label{thm:md1} It holds that
\begin{align*}
P_{\text{md}}\left(  \xi\right)   &  \leq\mathbb{P}_{\overline{\text{RIP}}}\\
&  +F\left(  4\xi\right) \\
&  +(4\tau^{2})^{k_{s}}n^{-k_{s}}\operatorname{SNR}^{-k_{s}}B_{1}\left(
m,k_{s}\right)  ,
\end{align*}
where
\begin{equation}
B_{1}\left(  m,k_{s}\right)  :=\sum_{j=0}^{m-1}\frac{\Gamma\left(
k_{s}+j\right)  }{\Gamma\left(  k_{s}\right)  j!} .\label{eq:b1}%
\end{equation}

\end{theorem}

The proof is deferred to the Appendix. Obviously, both, block correlation and
the HiHTP/HiIHT detector, achieve diversity order of $k_{s}$ but differ in the
\textquotedblleft shifts\textquotedblright\ $B_{0}$ and $B_{1}$. Note that a
numerical evaluation of the expression \eqref{eq:b1} for $B_{1}$ is much more
tractable than the expression \eqref{eq:b0} of $B_{0}$. Moreover, it can be
readily shown that an explicit formula is given by \cite{Wunder2017_ASILOMAR}%
\[
\sum_{j=0}^{m-1}\frac{\Gamma\left(  k_{s}+j\right)  }{\Gamma\left(
k_{s}\right)  j!}=\frac{m}{k_{s}}\binom{m+k_{s}-1}{k_{s}-1}.
\]
Notably, $F(\xi)$ can be numerically evaluted. But since $\xi$ is small, we
can as well use the approximations in the appendix, so that%
\[
F(\xi)\leq\int_{0}^{\xi}f_{X}\left(  x\right)  dx=\frac{\xi^{k_{s}}}%
{k_{s}\Gamma\left(  k_{s}\right)  }+o\left(  \xi^{k_{s}}\right) .
\]
Hence, consequently, we may roughly select $\xi\in O\left(  \operatorname{SNR}%
^{-1}\right)  $ to be negligible with respect to the diversity term.

\paragraph*{Comparison of the asymptotics}

We can now compare the asymptotics of $B_{0}\left(  s,u\right)  $ and
$B_{1}\left(  m,k_{s}\right)  $ where $m=m\left(  s,u,k_{s},k_{u}\right)  $.
To this end, we fix $k_{u}$ and $k_{s}$ and let either $u$ or $s$ or both
grow. For the classical correlator in (\ref{eq:corr}), a misdetection event
occurs if the collected noise energy is larger than the collected channel path
energy within the support of the cyclic shifts of some fixed active user
signature. Define $e_{j}^{\left(  c\right)  }\sim\mathcal{CN}\left(
0,\sigma^{2}I_{s}\right)  $, where $e^{\left(  c\right)  }$ is of size $s$,
and assume orthogonal signatures, i.e. $D\left(  h\right)  $ is a unitary
matrix. Since the $s$ noise terms are incoherently added whereas each channel
path scales with the signature's energy, we get from the results in \cite[eqn.
(18)]{Lee2013_TransWC} that the missed detection $P_{\text{md}}$ for $\xi=0$
is given by%
\begin{align*}
&  \mathbb{P}\left(  \left\{  \max\nolimits_{0\leq j<u-1,,j\neq i}\Vert
e_{j}^{\left(  c\right)  }\sqrt{n}\Vert^{2}>\Vert nh_{i}+e_{i}^{\left(
c\right)  }\sqrt{n}\Vert^{2}\right\}  \right)  \\
&  \geq\mathbb{P}\left(  \left\{  \max\nolimits_{0\leq j<u-1}\Vert
e_{j}^{\left(  c\right)  }\sqrt{n}\Vert^{2}>\left\Vert nh_{i}\right\Vert
^{2}\right\}  \right)  \\
&  \geq n^{-k_{s}}\operatorname{SNR}^{-k_{s}}B_{1}\left(  s,k_{s}\right)  ,
\end{align*}
where the third step holds for large enough SNR. On the other hand, we have
from the proof of Theorem \ref{thm:md1} the upper bound%
\begin{align*}
&  \mathbb{P}\left(  \left\{  \frac{4\tau^{2}}{n}\lVert e\rVert^{2}%
\geq\left\Vert h_{i}\right\Vert ^{2}\right\}  \right)  \\
&  \leq n^{-k_{s}}\operatorname{SNR}^{-k_{s}}(4\tau^{2})^{k_{s}}B_{1}\left(
m,k_{s}\right)  ,
\end{align*}
again for large SNR. Here, $m$ is of the order $k_{u}\log({u}/{k}_{u}%
)+k_{u}k_{s}\log({s}/k_{s})$ implying $\mathbb{P}_{\overline{\text{RIP}}%
}\rightarrow0$ and some finite noise enhancement $\tau<\infty$ for large $s$,
see Theorem\ref{thm:HiHTP}. In particular, $m$ grows only logarithmically and
not linear in $s$. Therefore, in the limit of large $s$ and fixed $u$, which
implies large $n$, the scaling of the bound is exponentially slower for
HiHTP/HiIHT compared to the classical block correlation detector. Thus, in
this regime of large system designs HiHTP/HiIHT is expected to outperform the
classical block correlation detector. We observe that the same finding is true
for any sub-linear scaling of $k_{s}$.

A similar comparison to the classical block correlation detector can be made
if $u$ grows while $s$ is constant. Since (\ref{eq:b0}) is too complicated to
be directly analyzed, we apply the union bound in Prop.~\ref{prop:max}, which
is given in appendix, and find that%
\begin{align}
&  \mathbb{P}\left(  \left\{  \max\nolimits_{0<j<u-1}\|e_{j}^{\left(
c\right)  }\sqrt{n}\|>\|nh_{i}+e_{i}^{\left(  c\right)  }\sqrt{n}\|\right\}
\right) \label{eq:lb_c}\\
&  \geq\mathbb{P}\left(  \left\{  \max\nolimits_{0<j<u-1}\|e_{j}^{\left(
c\right)  }\sqrt{n}\|>\|nh_{i}\|\right\}  \right) \nonumber\\
&  \geq n^{-k_{s}}\operatorname{SNR}^{-k_{s}}\left(  u-1\right)  B_{1}\left(
s,k_{s}\right)  +o\left(  \operatorname{SNR}^{-k_{s}}\right) .\nonumber
\end{align}
Hence, again, since $m$ and $B_{1}\left(  m,k_{s}\right)  $ grows only
logarithmically and not linear in $u$, and since $B_{1}\left(  m,k_{s}\right)
\simeq\frac{m}{k_{s}}\left(  \frac{e\left(  m+k_{s}-1\right)  }{k_{s}%
-1}\right)  ^{k_{s}-1}$, in the limit of large $u$ and sub-linear scaling of
$k_{u}$, the scaling of the bound is exponentially slower for HiHTP/HiIHT.
Altogether, we conclude:

\begin{theorem}
\label{thm:asymp} Under the signatures' choice of \eqref{eq:signa}, the
hierarchical thresholding detector outperforms the classical block correlator
with respect to $P_{\text{md}}$ in the regime of large $n=us$ and only
sub-linear scaling (in $n$) of $k_{u}k_{s}$.
\end{theorem}

While these asymptotics justify the application of HiHTP/HiIHT in the sparse
setting there is some very unsatisfying properties of our upper bound.
Clearly, the upper bound depends on the noise enhancement parameter of
HiHTP/HiIHT $\tau$. In the proofs the noise enhancement typically is
conservatively estimated. Hence, a validation of bounds can only be done on
qualitative level. Another main problem is that, as we will see in the
simulations, it does not reflect the fact that user detection is actually
independent of $k_{u}$ as long as there is a sufficient number of
measurements. {This leads us to an alternative approach in the next section.}

\subsection{Improved user detection analysis}

Another approach {to bound the missed detection probability focuses on the
first linear estimation step of the support in the HiHTP/HiIHT
algorithm\footnote{It is easy to see that the {derivations hold for any linear
estimator $\Psi=A^{H}B$ where }$B$ is any positive semidefinite matrix (i.e.
it has a square root).}. Let $\Psi=A^{H}$ be the linear estimator used by
HiHTP/HiIHT and consider the noiseless case. Furthermore, we assume that the
signal strength of each block is bounded by $h_{\text{min}}\ \leq\Vert
h_{i}\Vert^{2}\leq h_{\text{max}}$. }Let us introduce $v_{1},\ldots,v_{k_{s}}$
as $k_{s}$ vectors of the sparsifying basis in $\mathbb{C}^{us}$. With the
help of this basis we can write
\[
\left\Vert h_{i}\right\Vert ^{2}=\sum_{j\in\mathcal{A}_{i}}\left\vert
\left\langle h,v_{j}\right\rangle \right\vert ^{2}.
\]
By assumption, if the user is active, i.e. $i\in\mathcal{A}^{B}$ we have
\[
\sum_{j\in\mathcal{A}_{i}}\left\vert \left\langle h,v_{j}\right\rangle
\right\vert ^{2}\geq h_{\text{min}},
\]
while if the user is inactive, i.e.\ $i\in\overline{\mathcal{A}^{B}}$, it is
\[
\sum_{j\in\mathcal{A}_{i}^{\ast}}\left\vert \left\langle h,v_{j}\right\rangle
\right\vert ^{2}=0.
\]
Suppose the energy threshold $\xi$ is chosen as $0<\xi<h_{\min}$ and define
$\epsilon\coloneqq\min\left\{  \xi,h_{\min}-\xi\right\}  $. We denote the set
of all possible index sets of cardinality $k_{s}$ and indices only in the
$i$-th block by $\Omega_{i}$. The thresholding operator $L_{\xi}^{B}$ applied
to the linear estimation $\Psi y$ does identify the correct set of users if
the following condition holds:
\begin{equation}
\max_{i\in\left[  u\right]  ,\omega\in\Omega_{i}}\left\vert \sum_{j\in\omega
}\left\vert \left\langle h,v_{j}\right\rangle \right\vert ^{2}-\sum
_{j\in\omega}\left\vert \left\langle \Psi y,v_{j}\right\rangle \right\vert
^{2}\right\vert \leq\epsilon. \label{eq:e_set}%
\end{equation}
We denote the probability that this condition~\eqref{eq:e_set} does not hold
for a given $h$ by $\mathbb{P}_{\overline{\text{sRIP}}}(\epsilon\mid h)$.

In fact, the condition~\eqref{eq:e_set} implies that for some estimated set
$\mathcal{A}^{\ast}$ that the linear estimator identifies%
\begin{align*}
&  \sum_{j\in\mathcal{A}^{\ast}_{i}}\left\vert \left\langle \Psi
y,v_{j}\right\rangle \right\vert ^{2}\\
&  \geq\sum_{j\in\mathcal{A}_{i}}\left\vert \left\langle \Psi y,v_{j}%
\right\rangle \right\vert ^{2}\\
&  \geq\sum_{j\in\mathcal{A}_{i}}\left\vert \left\langle h,v_{j}\right\rangle
\right\vert ^{2}-\epsilon\\
&  >\xi,\;\text{if }i\in\mathcal{A}^{B}%
\end{align*}
and
\[
\sum_{j\in{\mathcal{A}^{\ast}}_{i}}\left\vert \left\langle \Psi y,v_{j}%
\right\rangle \right\vert ^{2}\leq\epsilon\leq\xi,\;\text{if }i\in
\overline{\mathcal{A} ^{B}}%
\]
since $\sum_{j\in\Omega_{i}}\left\vert \left\langle h,v_{j}\right\rangle
\right\vert ^{2}=0,$ if $i\in\overline{\mathcal{A} ^{B}}$, so that the block
is correctly detected (without noise).

In the following lemma we will show that the condition~\eqref{eq:e_set} holds
with high probability for sufficiently large $m$ on average over all $h$. We
will use the model (\ref{eq:2edModel}) for the signals $h$ and incorporate the
noise according to this model.

\begin{lemma}
\label{lem:FFTconc} Let $\epsilon>0$. Then, the event (\ref{eq:e_set}) holds
with probability%
\begin{align*}
&  \mathbb{P}_{\overline{\text{sRIP}}}\left(  \epsilon\right)  \\
&  \leq32u\left(  \frac{es}{k_{s}}\right)  ^{k_{s}}k_{s}e^{-\frac{\epsilon
^{2}m}{O\left(  4k_{u}^{4}k_{s}^{5}\right)  }}\\
&  +\frac{32k_{s}}{\sqrt{2}}\frac{\operatorname{SNR}^{-k_{u}k_{s}}}{n^{2k_{s}%
}}\left(  \frac{k_{u}^{4}k_{s}^{5}}{n^{(1-3/k_{u})}}\right)  ^{k_{u}k_{s}}\\
&  +4e^{-\frac{k_{u}k_{s}}{2}}%
\end{align*}
for sufficiently large SNR, $k_{u}\geq8,k_{s}\geq3$.
\end{lemma}

Several remarks are in order:

\begin{itemize}
\item The result is asymptotic in nature such that it holds for large SNR and
large $k_{u}k_{s}$ (with typically fixed $k_{s}$) and hence large $n\geq us$
which reflects the mMTC scenario. Specifically, we see that for any fixed
$\epsilon>0$, $\mathbb{P}_{\overline{\text{sRIP}}}\left(  \epsilon\right)
\rightarrow0$ provided that $m\geq m_{0}$ with
\begin{equation}
m_{0}\in O\left(  k_{u}^{4}k_{s}^{6}\log\left(  n\right)  \right)
\label{eq:m-scaling}%
\end{equation}
in the regime of large $n$ and any sub-linear scaling of $k_{u}k_{s}$ with
respect to $n^{(1-3/k_{u})}$, i.e. $\frac{k_{u}^{4}k_{s}^{5}}{n^{(1-3/k_{u})}%
}\rightarrow0$ for any $k_{u}k_{s},n\rightarrow\infty$. Within this parameter
regime channel energies are approximately recovered with small error
$\epsilon$. This will result in a $P_{\text{md}}(\xi)$ which is actually
independent of $k_{u}$ (see Theorem \ref{thm:md2}) and is validated in the
simulation section.

\item It is important to note that we neglect the cases where HiHTP/HiIHT will
move away in the iterations from an initially correct user detection (which
very rarely happens in the simulations). This is reasonable because
asymptotically for high SNR (which we target) this probability tends to zero
anyway since HiHTP/HiIHT will clearly output the correct result provided the
conditions of \ref{thm:HiHTP} are fulfilled as well.

\item The required number of measurements is much less than $n$ which results
in a considerable gain of user capacity since we may exploit the remaining
unused subcarriers to implement \emph{further parallel pilot channels}.
Although this result is asymptotically in nature, the user capacity gain is
also clearly apparent in the simulations.

\item Parameter $\epsilon$ is to be selected such that $F\left(
\epsilon\right)  $ is sufficiently small and depends only on the channel statistics.
\end{itemize}

Let us now turn to the user detection where we set $F^{z}(\xi):=\mathbb{P}%
(\Vert(h_{i}+z_{i}^{\prime})_{\mathcal{A}_{i}}\|^{2}\leq\xi)$.

\begin{theorem}
\label{thm:md2} It holds that
\begin{align*}
P_{\text{md}}\left(  \xi\right)   &  \leq\mathbb{P}_{\overline{\text{sRIP}}%
}\left(  \epsilon\right) \\
&  +2s\left(  u-k_{u}\right)  F^{z}(\xi+4\epsilon)\\
&  +n^{-k_{s}}\operatorname{SNR}^{-k_{s}}\left(  \frac{us}{2k_{s}m}\right)
^{-k_{s}}s\left(  u-k_{u}\right) .
\end{align*}

\end{theorem}

We see that the diversity order is the same as in Theorem \ref{thm:md1}
although much less measurements are required. Note that, for this result to be
meaningful, we need to establish that by virtue of eq. (\ref{eq:m-scaling})
the SNR scaling of Theorem \ref{thm:md2} is worse than the SNR scaling in
Lemma \ref{lem:FFTconc}. We have
\begin{align*}
&  n^{-k_{s}}\operatorname{SNR}^{-k_{s}}\left(  \frac{us}{2k_{s}m}\right)
^{-k_{s}}s\left(  u-k_{u}\right)  \\
&  \geq n^{-2k_{s}}\operatorname{SNR}^{-k_{s}}\left(  2k_{u}^{4}k_{s}^{7}%
\log\left(  us\right)  \right)  ^{k_{s}}s\left(  u-k_{u}\right)  \\
&  \geq n^{-2k_{s}}\operatorname{SNR}^{-k_{s}}k_{s}\left(  u-k_{u}\right)  ,
\end{align*}
which is, technically for $u-k_{u}\geq32/\sqrt{2}$, clearly much larger for
any SNR point and also falls much slower in SNR than the SNR dependent term in
Lemma \ref{lem:FFTconc} so that $P_{\text{md}}$ is indeed dominated by
$n^{-k_{s}}\operatorname{SNR}^{-k_{s}}\left(  \frac{us}{2k_{s}m}\right)
^{-k_{s}}s\left(  u-k_{u}\right)  $.

We will validate now our results in the next section.

\section{Evaluations and Simulations}

\label{sec:evaluation}

In the simulations we used HiIHT since it is faster. We tested the performance
of HiIHT in numerical simulations using the system parameters $n=1024$, $1\leq
u\leq16$, $1\leq s\leq256$. The size of the observation window was taken to be
$1\leq m\leq300$. For simplicity we assume that there is \emph{exactly} an
$(k_{u},k_{s})$-sparse multiuser channel so that we can set $\xi=0$ for the
sake of exposition. All performance metrics are user-wise and clearly this
performance does not depend on which users are actually active.

\begin{figure2}
[ptb]
\centering\includegraphics[width=2\columnwidth]{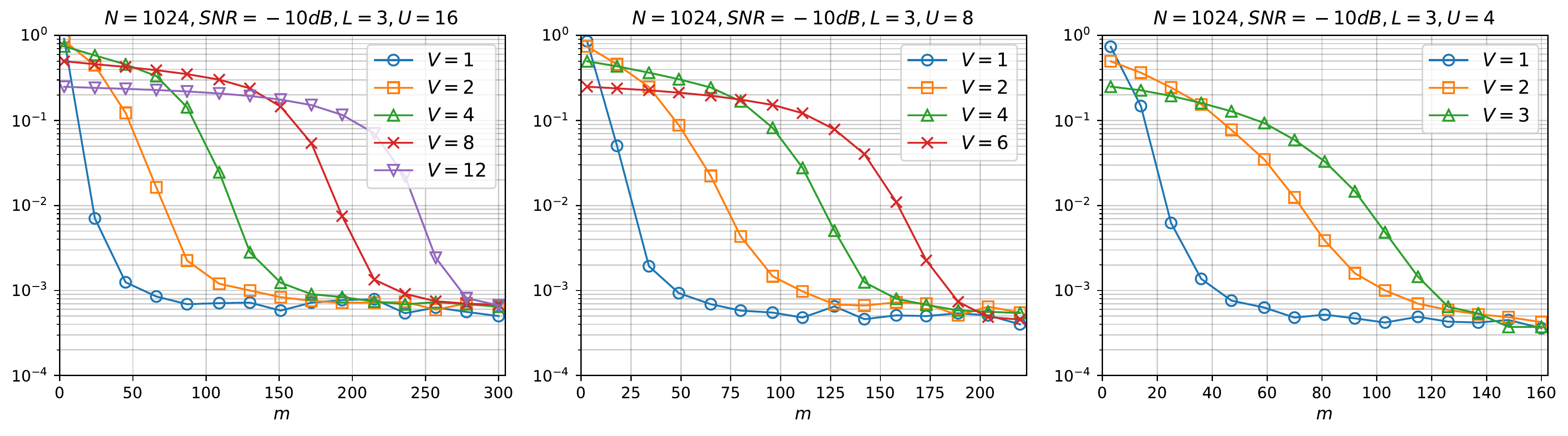}\caption{Average
MSE$_{i}$ for the active users dependent on the number of measurements $m$ for
$n=1024$, SNR$=-10$dB, and $L=3$. The user number from left to right is
$U=\{16,8,4\}$.}\label{fig:mDep}
\end{figure2}

Our first simulation in Fig. \ref{fig:mDep} shows the dependence on $m$ where
we depict the unnormalized frequency domain MSE$_{i}$ per subcarrier for the
active users \emph{averaged} over the runs, see eq. (\ref{eq:mse}). We see
that the 'phase transitions' occur for $m$ far less than $n$ as expected. From
this simulation we obtain that $m=300$ is sufficient for the targeted
parameter regime.

In Fig. \ref{fig:L3U4} to Fig. \ref{fig:L6U16} we simulated the user detection
performance. Obviously, simulations and upper bounds coincide quite nicely and
that in all cases the slope, i.e. diversity gain, is correctly represented
justifying our approach for the analysis. Generally, the bounds qualitatively
reflect the dependence on the system parameter in all cases provided $m$ is
selected sufficiently large such that HiIHT operates far beyond the phase
transition. This is clearly visible in Fig.'s \ref{fig:L6U8}, \ref{fig:L6U16}
where the performance is already worse for larger delay spread and even turns
into an error floor for $k_{u}=12$ in Fig. \ref{fig:L6U16}. There is a small
gap due to the union bound approach and the parameter $\epsilon$ which we
never make explicit. Hence, generally, the parameter setting has to be
carefully selected for the theorems to hold.

The user capacity gain is clearly visible in the simulation and is more than
300\% since $m/n\leq1/3$ which is a promising result. We have not incorporated
the block correlator since it will operate in the order of our upper bounds as
shown in \cite{Lee2013_TransWC}. This implies that we have not come close to
the asymptotic regime where the hierarchical thresholding algorithm
outperforms the block correlation detector.

\begin{figure}[ptb]
\centering\includegraphics[width=\columnwidth]{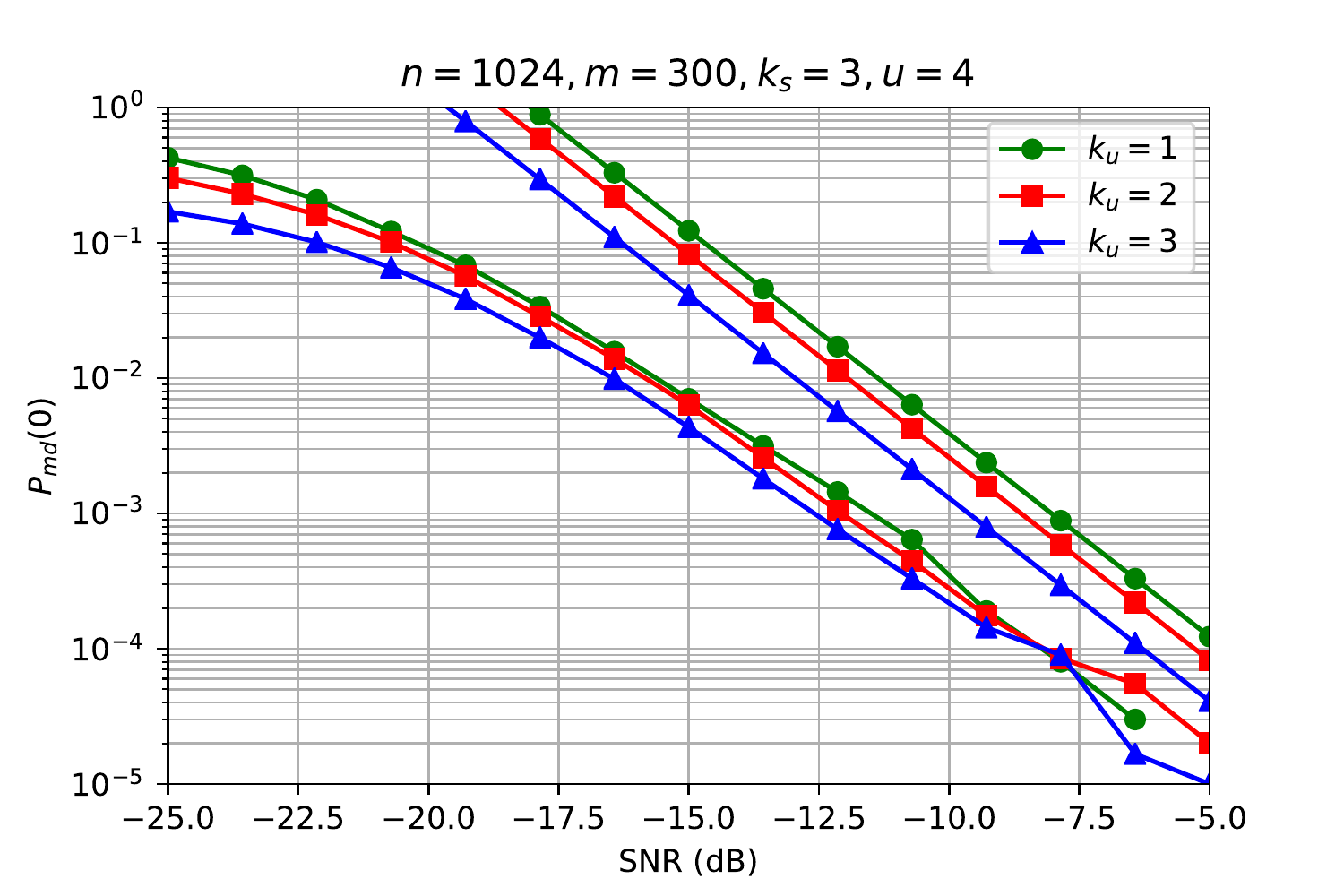}\caption{The
Figure depicts $P_{\text{md}}$ over SNR for $n=1024,m=300,k_{s}=3,u=4.$}%
\label{fig:L3U4}%
\end{figure}

\begin{figure}[ptb]
\centering\includegraphics[width=\columnwidth]{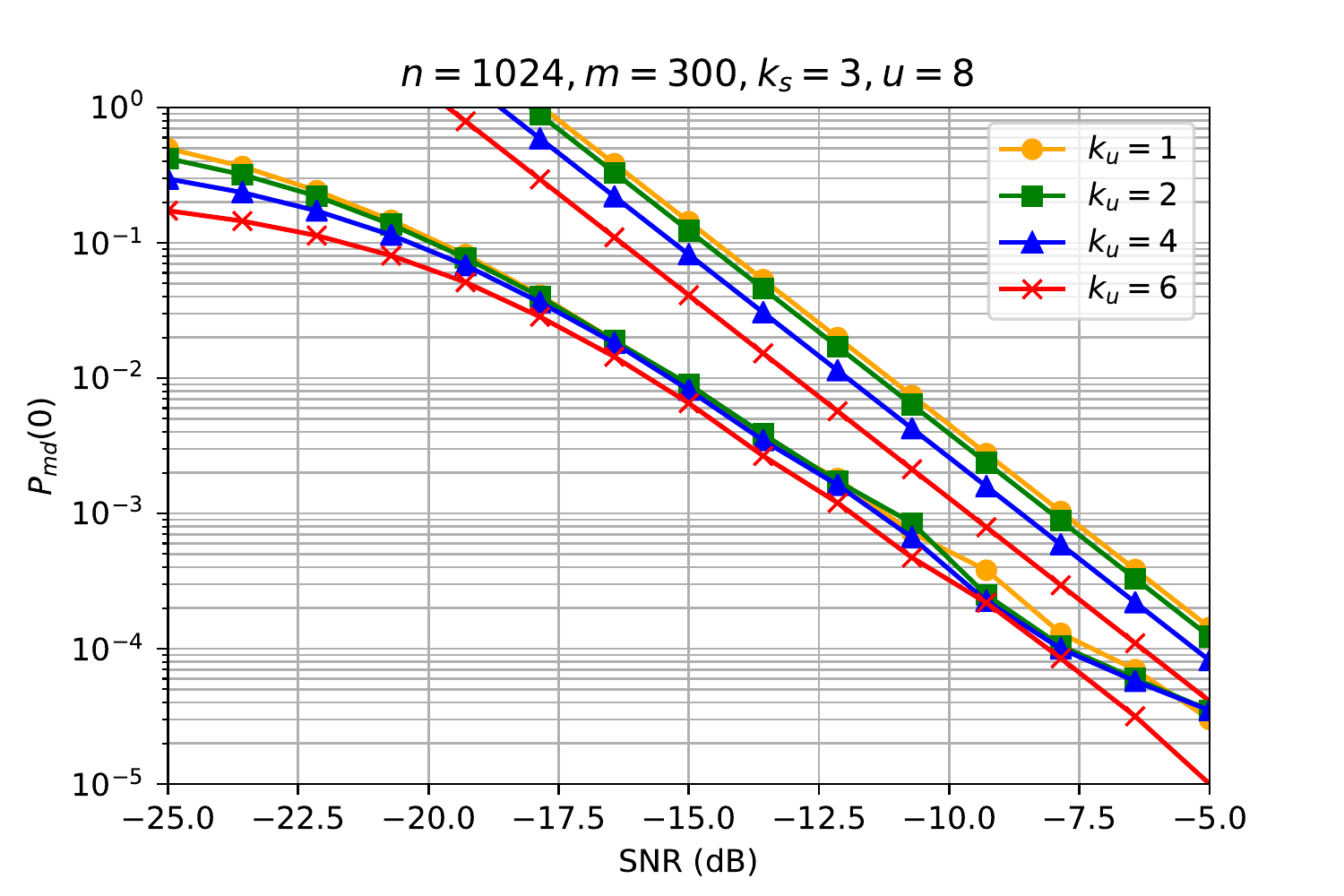}\caption{The
Figure depicts $P_{\text{md}}$ over SNR for $n=1024,m=300,k_{s}=3,u=8.$}%
\label{fig:L3U8}%
\end{figure}

\begin{figure}[ptb]
\centering\includegraphics[width=\columnwidth]{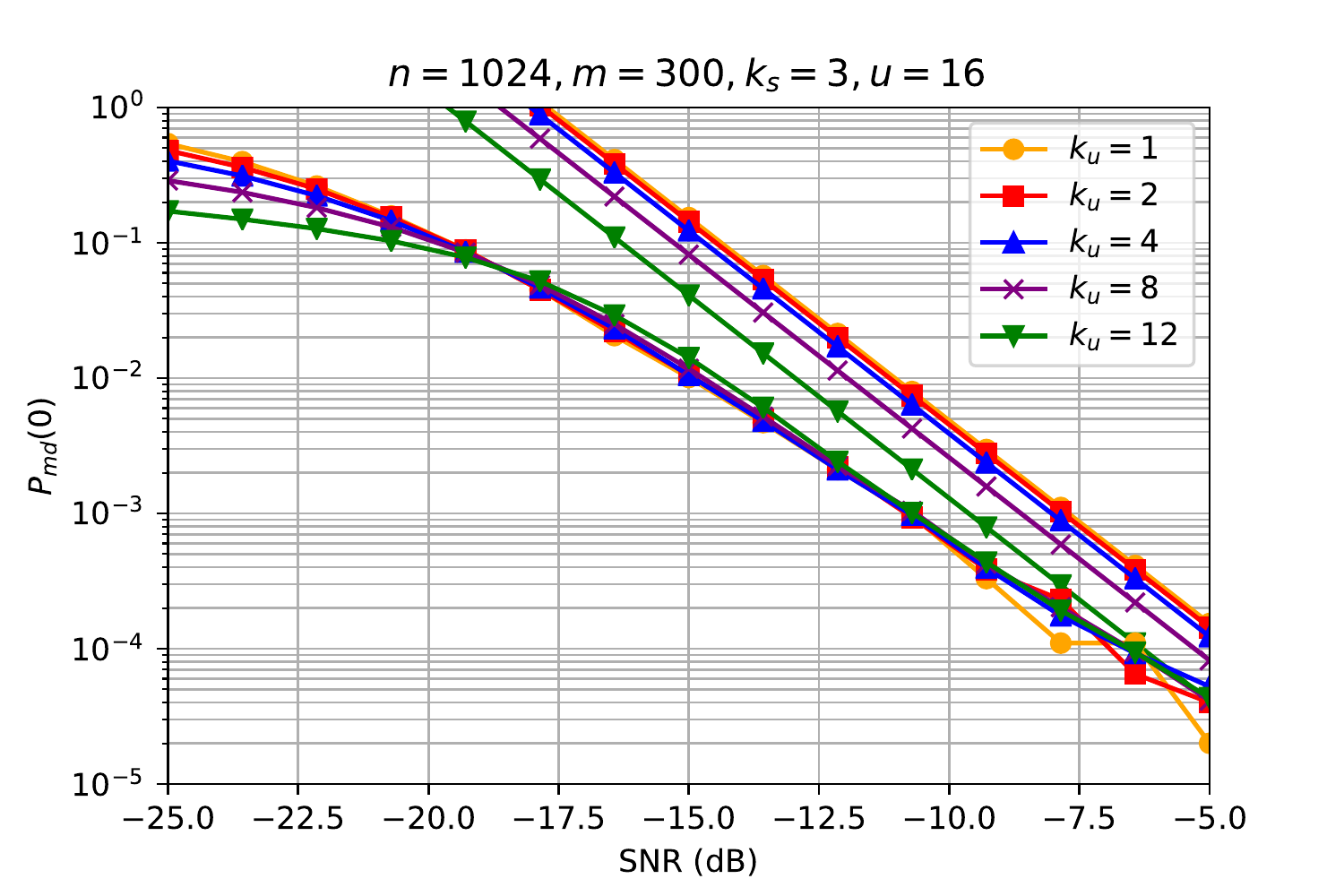}\caption{The
Figure depicts $P_{\text{md}}$ over SNR for $n=1024,m=300,k_{s}=3,u=16.$}%
\label{fig:L3U16}%
\end{figure}

\begin{figure}[ptb]
\centering\includegraphics[width=\columnwidth]{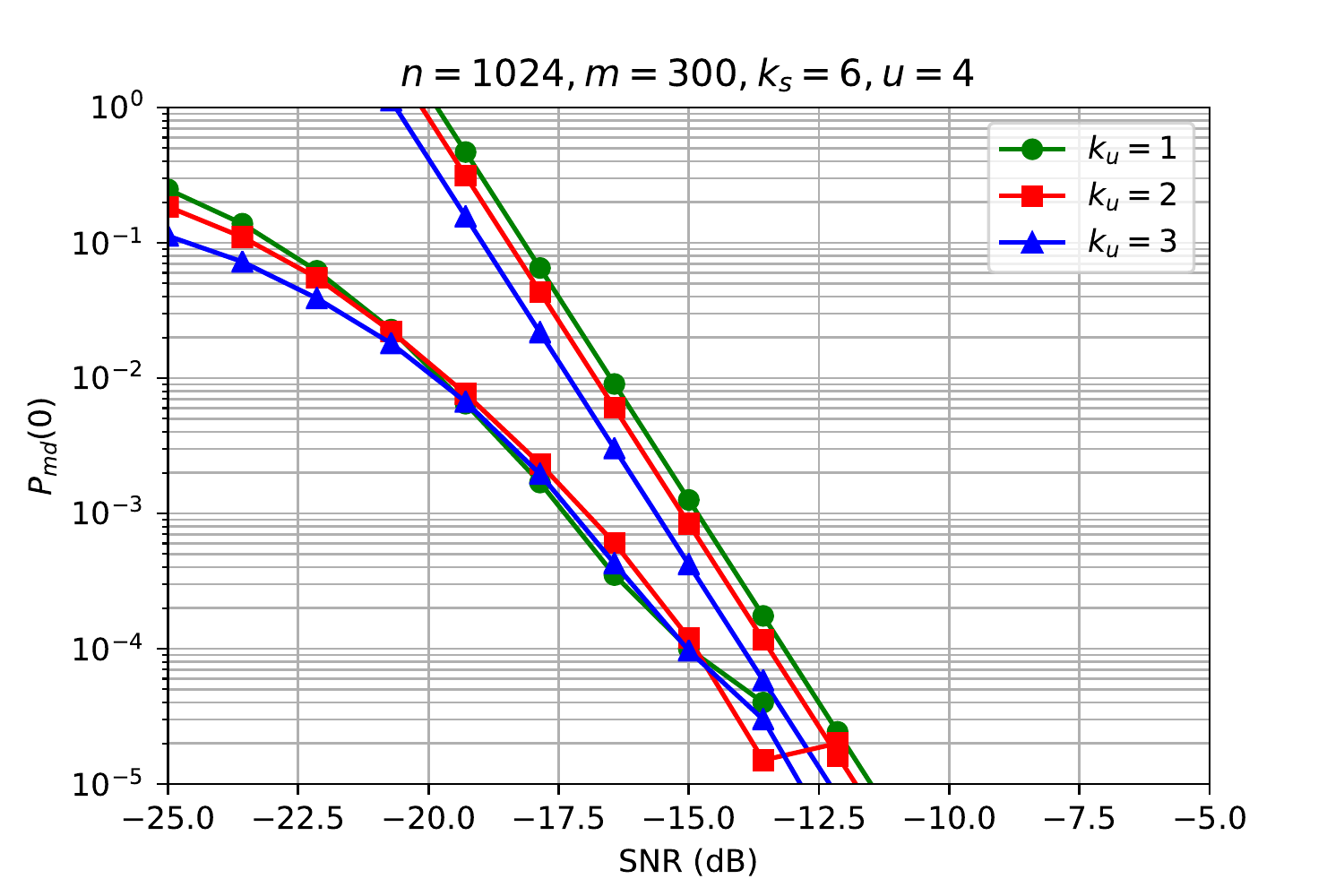}\caption{The
Figure depicts $P_{\text{md}}$ over SNR for $n=1024,m=300,k_{s}=6,u=4.$}%
\label{fig:L6U4}%
\end{figure}

\begin{figure}[ptb]
\centering\includegraphics[width=\columnwidth]{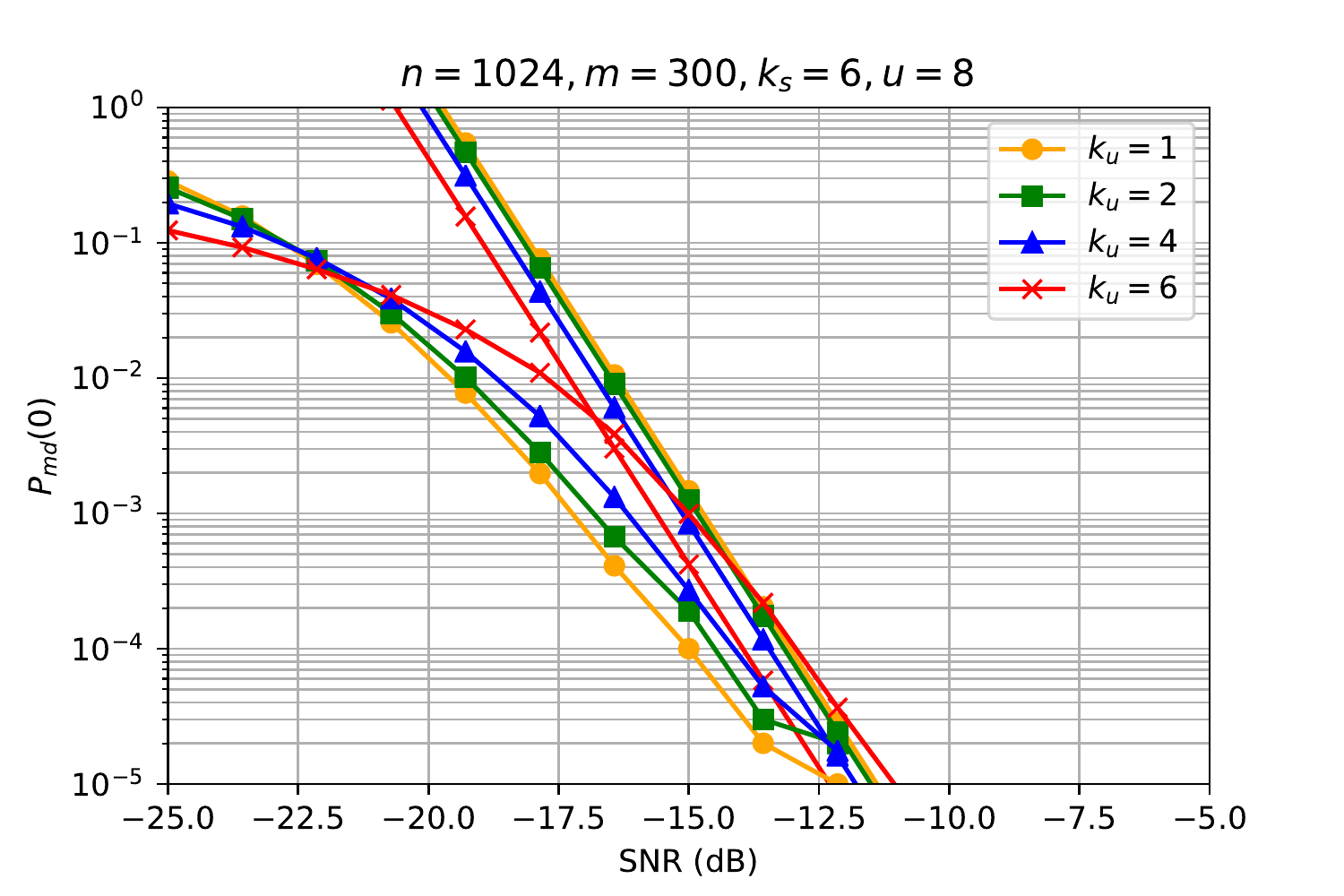}\caption{The
Figure depicts $P_{\text{md}}$ over SNR for $n=1024,m=300,k_{s}=6,u=8.$}%
\label{fig:L6U8}%
\end{figure}

\begin{figure}[ptb]
\centering\includegraphics[width=\columnwidth]{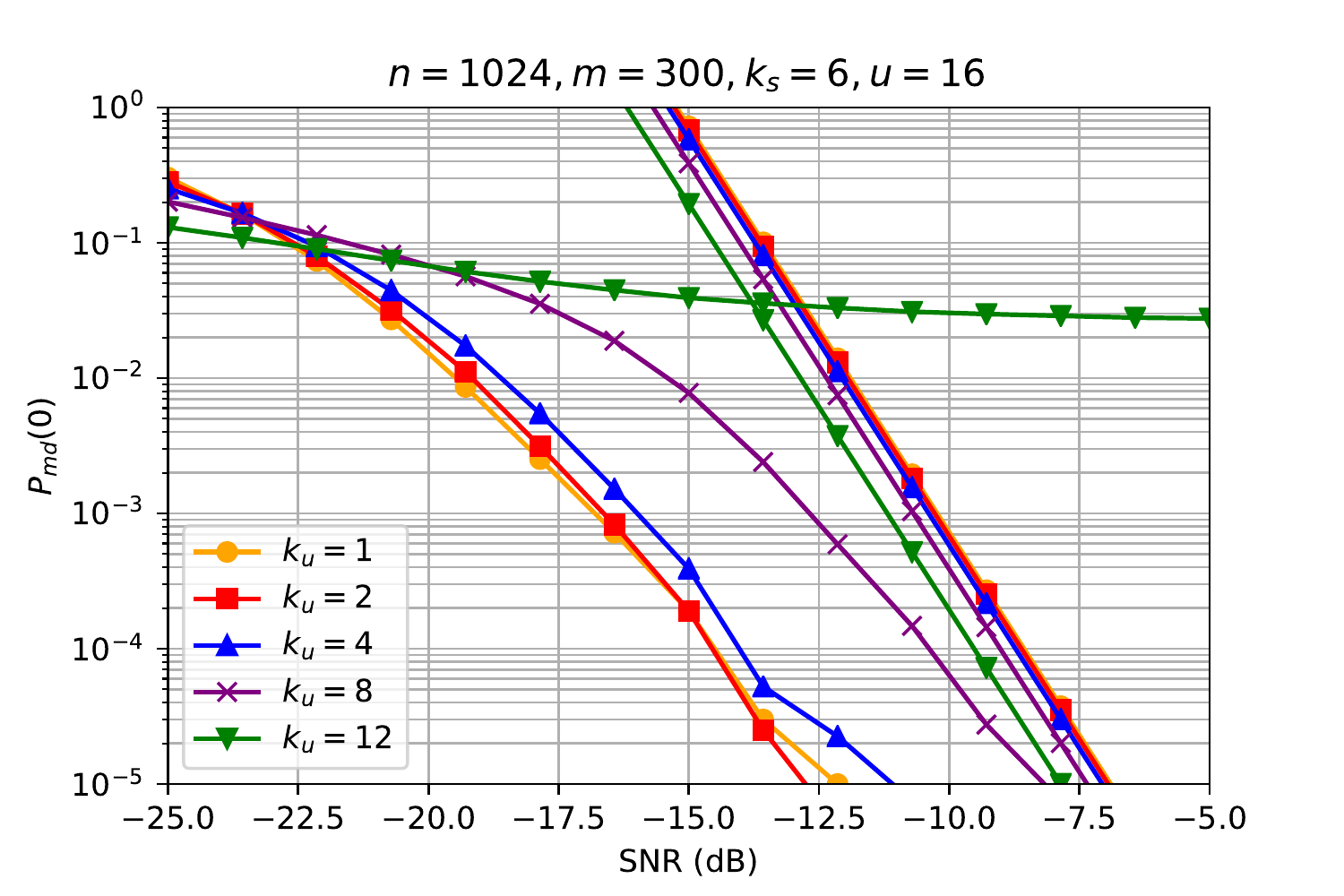}\caption{The
Figure depicts $P_{\text{md}}$ over SNR for $n=1024,m=300,k_{s}=6,u=16.$}%
\label{fig:L6U16}%
\end{figure}In Fig. \ref{fig:MSE} we evaluate finally the corresponding MSE
performance. Also here we see that the performance of the detector is
sufficient to obtain qualitatively good channel estimates in 'one shot' random
access scenarios.

\begin{figure2}
\centering\includegraphics[width=2\columnwidth]{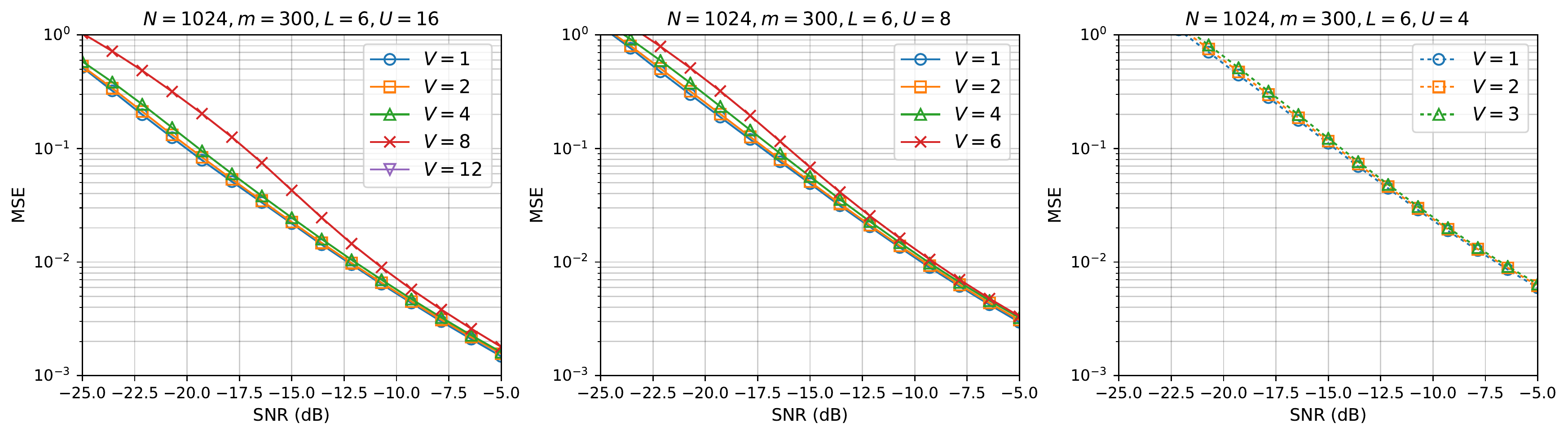}\caption{Average
MSE$_{i}$ dependent on the number of measurements $m$.}\label{fig:MSE}
\end{figure2}

\section{Conclusions and Outlook}

In this paper, the performance of hierarchical thresholding algorithms under
noise is studied for certain indicators, in particular, the block error and
missed detection probability. We provide upper bounds on the missed detection
probability in terms of the diversity order and relate them to the classical
block correlation detector. Our findings are that in a certain sparse
parameter regime the HiHTP detector can outperform the classical detection
schemes. This parameter regime is expected to arise ubiquitously in mMTC scenarios.

Very recently, in a series of papers, mMTC system design is combined with
\emph{massive MIMO} which adds another design parameter (number of antennas)
to the problem. In recent work \cite{Liu2018_ARXIV-1}, the Approximate Message
Passing (AMP) algorithm is considered for the demanding setting where sparsity
is growing linearly with the dimensions showing that detection probability can
be driven to zero with large number of antennas. On the other hand frequency
diversity is not considered in \cite{Liu2018_ARXIV-1}. Another work can be
found in \cite{Carvalho2017_TWC} where approximative analysis is provided without considering multipath. Hence, it would be interesting to see the
effect of multipath in our work in \cite{Wunder2018_ARXIV-2} where
error rates tend to decrease with the same diversity order.

\section{Acknowledgements}

This work was carried out within DFG grants WU 598/7-1 and WU 598/8-1 (DFG
Priority Program on Compressed Sensing). The research of IR and JE has been
supported by the DFG project EI 519/9-1 (SPP1798 CoSIP), the Templeton
Foundation, the EU (RAQUEL), the BMBF (Q.com), and the ERC (TAQ).

\appendix

\subsection{Proof of Theorem \ref{thm:md1}}

We need the following proposition for any of the following theorems. The
following proposition is the 'high SNR' probability approximation of
misdetection events which is used in Theorem \ref{thm:md2} and Proposition
\ref{prop:FFTconc}.

\begin{proposition}
\label{prop:max} Let $z_{0},\ldots,z_{u-1}\sim\mathcal{CN}\left(  0,\sigma
^{2}I_{m}\right)  $ and $(h_{i})_{\mathcal{A}_{1}}\sim\mathcal{CN}\left(
0,\operatorname{diag}\left(  \sigma_{h,0}^{2},\sigma_{h,1}^{2},\ldots
,\sigma_{h,k_{s}-1}^{2}\right)  \right)  $ with $\mathcal{A}\subset\lbrack s]$
some index set of size $|\mathcal{A}|=k_{s}<s$. Then%
\begin{align*}
&  \mathbb{P}\left(  \max\nolimits_{j\neq i}\|z_{j}\|^{2}\geq\Vert
(h_{i})_{\mathcal{A}_{1}}\|^{2}\right) \\
&  =\frac{u\sigma^{2k_{s}}}{\prod_{i=0}^{k_{s}-1}\sigma_{h,i}^{2}}\sum
_{j=0}^{m-1}\frac{\Gamma\left(  k_{s}+j\right)  }{\Gamma\left(  k_{s}\right)
j!}+o\left(  \operatorname{SNR}^{k_{s}}\right) .
\end{align*}

\end{proposition}

\begin{proof}
Fix $i=1$. Set $\Vert\left(  h_{1}\right)  _{\mathcal{A}}\Vert^{2}=:X$.
Clearly, by the union bound
\[
\mathbb{P}\left(  \max\nolimits_{j>1}\Vert z_{j}\Vert^{2}>X\right)
\leq\left(  u-1\right)  \mathbb{P}\left(  \Vert z_{2}\Vert^{2}>X\right)
\]
and conversely%
\begin{align*}
&  \mathbb{P}\left(  \max\nolimits_{j>1}\Vert z_{j}\Vert^{2}>X\right)  \\
&  =1-\left(  1-\mathbb{P}\left(  \Vert z_{2}\Vert^{2}>X\right)  \right)
^{u-1}\\
&  =1-\exp\left(  \left(  u-1\right)  \log\left(  1-\mathbb{P}\left(  \Vert
z_{2}\Vert^{2}>X\right)  \right)  \right)  \\
&  \geq1-\exp\left(  -\left(  u-1\right)  \mathbb{P}\left(  \Vert z_{2}%
\Vert^{2}>X\right)  \right)  \\
&  \geq\left(  u-1\right)  \mathbb{P}\left(  \Vert z_{2}\Vert^{2}>X\right)  -%
%TCIMACRO{\tsum _{i=2}^{\infty}}%
%BeginExpansion
{\textstyle\sum_{i=2}^{\infty}}
%EndExpansion
\frac{u^{i}}{i!}\mathbb{P}^{i}\left(  \Vert z_{2}\Vert^{2}>X\right)  \\
&  \geq\left(  u-1\right)  \mathbb{P}\left(  \Vert z_{2}\Vert^{2}>X\right)
-\mathbb{P}^{2}\left(  \Vert z_{2}\Vert^{2}>X\right)
%TCIMACRO{\tsum _{i=2}^{\infty}}%
%BeginExpansion
{\textstyle\sum_{i=2}^{\infty}}
%EndExpansion
\frac{u^{i}}{i!}\\
&  \geq\left(  u-1\right)  \mathbb{P}\left(  \Vert z_{2}\Vert^{2}>X\right)
-\mathbb{P}^{2}\left(  \Vert z_{2}\Vert^{2}>X\right)  e^{u}.
\end{align*}
The first inequality follows from $\log\left(  1-x\right)  \leq-x$ and
$\exp\left(  x\right)  $ is increasing over $\mathbb{R}$. The second
inequality follows from expanding the exponential term. The third is due to
$\mathbb{P}(\Vert z_{2}\Vert^{2}>X)<1$. Hence, for fixed $u$ we have the
desired converse. To proceed, we denote the density of $X$ as $f_{X}$. We want
to calculate%
\begin{equation}
\mathbb{P}(\{\lVert z_{2}\rVert^{2}\geq X\})=\int\mathbb{P}(\lVert z_{2}%
\rVert^{2}\geq X|X=x)f_{X}(x)dx,.\label{eq:IntCdfPdf}%
\end{equation}
The complementary cumulative distribution function for the squared norm of the
complex Gaussian noise $z_{1}\sim\mathcal{CN}(0,\frac{\sigma^{2}}{n}I_{m})$ is
given by \cite{Lee2013_TransWC}:%
\begin{equation}
\mathbb{P}\left(  \lVert z_{2}\rVert^{2}>x\right)  =\sum_{j=0}^{m-1}%
\frac{x^{j}\exp\left(  -\frac{x}{\sigma^{2}}\right)  }{j!\sigma^{2j}%
}.\label{eq:Cdf}%
\end{equation}
Moreover, it was shown in \cite{Wang2003_TransComm} and applied in
\cite{Lee2013_TransWC}, that the probability density function of $\gamma$ can
be approximated around $\gamma=0$ in the high $\operatorname{SNR}$ regime by
\begin{equation}
f_{\gamma}\left(  x\right)  =\frac{\prod_{i=0}^{k_{s}-1}\sigma_{h,i}^{-2}%
}{\Gamma\left(  k_{s}\right)  }x^{k_{s}-1}+o\left(  x^{k_{s}}\right)
.\label{eq:Pdf}%
\end{equation}
Since the performance parameters depend only on the behaviour of the function
near the origin, the part $o\left(  x^{k_{s}}\right)  $ can be neglected in
the performance analysis for high $\operatorname{SNR}$ regimes. Plugging
(\ref{eq:Pdf}) and (\ref{eq:Cdf}) into (\ref{eq:IntCdfPdf}) yields:%
\begin{align*}
&  \int\mathbb{P}(\lVert z_{2}\rVert^{2}\geq X\mid X=x)f_{X}(x)dx\\
&  =\frac{\prod_{i=0}^{k_{s}-1}\sigma_{h,i}^{-2}}{\Gamma\left(  k_{s}\right)
}\sum_{j=0}^{m-1}\frac{1}{j!\sigma^{2j}}\int\limits_{0}^{\infty}x^{k_{s}%
-1}x^{j}\exp\left(  -\frac{x}{\sigma^{2}}\right)  dx\\
&  =\frac{\prod_{i=0}^{k_{s}-1}\sigma_{h,i}^{-2}}{\Gamma\left(  k_{s}\right)
}\sum_{j=0}^{m-1}\frac{\Gamma\left(  k_{s}+j\right)  }{j!\sigma^{2j}}\left(
\frac{1}{\sigma^{2}}\right)  ^{-k_{s}-j}\\
&  =\frac{\sigma^{2k_{s}}}{\prod_{i=0}^{k_{s}-1}\sigma_{i}^{2}}\sum
_{j=0}^{m-1}\frac{\Gamma\left(  k_{s}+j\right)  }{\Gamma\left(  k_{s}\right)
j!}%
\end{align*}
where we used that
\begin{equation}
\int_{0}^{\infty}x^{k-1}e^{-Ax}dx=\Gamma\left(  k\right)  A^{-k}%
.\label{eq:integral}%
\end{equation}
Observing that the result is independent of the actual support of $h_{i}$
completes the proof.
\end{proof}

We now provide a proof for the Theorem~\ref{thm:md1}. For the proof it is
sufficient to consider a single fixed user with index $i\in\mathcal{A}^{B}$
out of the atmost $k_{u}$ active users.

\begin{proof}
[Proof of Theorem \ref{thm:md1}]By the definition of the user detection scheme
using HiHTP, the detection of a user requires two conditions to be met: First,
the reconstructed block $\hbar_{i}$ must be larger in $\ell_{2}$-norm than all
blocks which were regarded as inactive. Second, the norm must exceed the
energy threshold of $\sqrt{\xi}$. Thus, for the active $i$-th user to be
missed at the same time at least one inactive user must be detected as active.
In other words, the $\ell_{2}$-norm of the $i$-th block of $\hbar$ is smaller
than at least one block of $\hbar$ with index in the complement $\overline
{\mathcal{A}^{B}}$ of $\mathcal{A}^{B}$. Here, the complement is taken with
respect to ${1,\ldots u}$. We therefore find that the event of
\textquotedblleft missed detection\textquotedblright\ of the $i$-th user is
included as follows:%
\begin{align*}
&  \left\{  \text{missed detection of user }i\right\} \\
&  \subset\left\{  \left\Vert \hbar_{i}\right\Vert \leq\sqrt{\xi}\right\}
\cup\left\{  \max\nolimits_{j\in\overline{\mathcal{A}^{B}}}\left\Vert
\hbar_{j}\right\Vert \geq\left\Vert \hbar_{i}\right\Vert \right\}  .
\end{align*}
Note that this inclusion is proper since the second event does not guarantee
that the $i$-th user is not detected.

Let us abbreviate the difference between the recovered and the original signal
by $d_{i}:=\hbar_{i}-h_{i}$. The first event then implies the following
conditions:
\begin{align*}
&  \left\{  \left\Vert \hbar_{i}\right\Vert \leq\sqrt{\xi}\right\}
\equiv\left\{  \left\Vert \hbar_{i}-h_{i}+h_{i}\right\Vert \leq\sqrt{\xi
}\right\} \\
&  \subseteq\left\{  \left\vert \left\Vert d_{i}\right\Vert -\left\Vert
h_{i}\right\Vert \right\vert \leq\sqrt{\xi}\right\} \\
&  \equiv\left\{  -\sqrt{\xi}\leq\left\Vert d_{i}\right\Vert -\left\Vert
h_{i}\right\Vert \leq\sqrt{\xi}\right\} \\
&  \subseteq\left\{  \left\Vert h_{i}\right\Vert -\left\Vert d_{i}\right\Vert
\leq\sqrt{\xi}\right\}  \equiv\left\{  \sqrt{\xi}+\left\Vert d_{i}\right\Vert
\geq\left\Vert h_{i}\right\Vert \right\} \\
&  \subseteq\left\{  \left\Vert d_{i}\right\Vert \geq\frac{1}{2}\left\Vert
h_{i}\right\Vert \right\}  \cup\left\{  \sqrt{\xi}\geq\frac{1}{2}\left\Vert
h_{i}\right\Vert \right\}
\end{align*}
where $\mathbb{P}(\{\sqrt{\xi}\geq\frac{1}{2}\|h_{i}\|\})=F(4\xi)$. The second
event requires that%
\begin{align*}
&  \left\{  \max\nolimits_{j\in\overline{\mathcal{A}^{B}}}\left\Vert \hbar
_{j}\right\Vert \geq\left\Vert \hbar_{i}\right\Vert \right\} \\
&  \subseteq\left\{  \max\nolimits_{j\in\overline{\mathcal{A}^{B}}}\left\Vert
\hbar_{j}-h_{j}\right\Vert \geq\left\Vert \hbar_{i}-h_{i}+h_{i}\right\Vert
\right\} \\
&  \subseteq\left\{  \max\nolimits_{j\in\lbrack u]}\left\Vert d_{j}\right\Vert
+\left\Vert d_{i}\right\Vert \geq\left\Vert h_{i}\right\Vert \right\} \\
&  \subseteq\left\{  \left\Vert d\right\Vert \geq\frac{1}{2}\left\Vert
h_{i}\right\Vert \right\} .
\end{align*}
After squaring both sides, we can utilize Theorem \ref{thm:HiHTP} to show that%
\[
\mathbb{P}\left(  \left\{  4\left\Vert d\right\Vert ^{2}\geq\left\Vert
h_{i}\right\Vert ^{2}\right\}  \right)  \leq\mathbb{P}_{\overline{\text{RIP}}%
}+\mathbb{P}\left(  \left\{  \frac{4\tau^{2}}{n}\lVert z\rVert^{2}%
\geq\left\Vert h_{i}\right\Vert ^{2}\right\}  \right) .
\]
Now we can invoke Prop. \ref{prop:max} to show that
\begin{align*}
&  \mathbb{P}\left(  \left\{  \frac{4\tau^{2}}{n}\lVert z\rVert^{2}%
\geq\left\Vert h_{i}\right\Vert ^{2}\right\}  \right) \\
&  \leq\left(  \frac{4\tau^{2}\sigma^{2}}{n\sigma_{h}^{2}}\right)  ^{-k_{s}%
}B_{1}\left(  m,k_{s}\right) \\
&  =\left(  4\tau^{2}\right)  ^{k_{s}}n^{-k_{s}}\operatorname{SNR}^{-k_{s}%
}B_{1}(m,k_{s})
\end{align*}
which concludes the proof.
\end{proof}

\subsection{Proof of Lemma \ref{lem:FFTconc}}

The proof of Lemma~\ref{lem:FFTconc} relies on concentration inequalities for
the measurement map $A$ as well as the norm of the signal $h$ and the noise
$z^{\prime}$ introduced as model~\eqref{eq:2edModel}.

Let us first introduce the norm concentration for a Gaussian linear mappings.
Let $G\in\mathbb{C}^{m\times n}$ be a random matrix with i.i.d.\ Gaussian
entries $\left(  G\right)  _{ki}\sim\mathcal{CN}(0,\frac{1}{m})$. The
normalisation ensures that $\mathbb{E}\left\Vert Ax\right\Vert ^{2}=\left\Vert
x\right\Vert ^{2}$ for any vector $x\in\mathbb{C}^{n}$. For a Gaussian random
matrix, it holds that \cite[Lemma 9.8]{Rahut2013}:
\begin{equation}
\mathbb{P}\left(  \left\vert \left\Vert Ax\right\Vert ^{2}-\left\Vert
x\right\Vert ^{2}\right\vert >\epsilon\left\Vert x\right\Vert ^{2}\right)
\leq2e^{-\frac{\epsilon^{2}m}{2}} \label{eq:Mconc1}%
\end{equation}
for every $\epsilon\in(0,1)$.

This concentration inequality allows us to also derive a concentration bound
for the the norm of a Gaussian random vector. To this end, we choose $m=n$ and
the vector $x_{\sigma}\coloneqq(\sigma,\sigma,\ldots,\sigma)^{T}\in
\mathbb{C}^{m}$. With this choice, we have that $x\coloneqq Ax_{\sigma}%
\sim\mathcal{CN}(0,\sigma^{2}I_{m})$ is a random vector with i.i.d.\ Gaussian
entries of variance $\sigma^{2}$. From eq.~\eqref{eq:Mconc1} we conclude that
\begin{align}
2e^{-\frac{\epsilon^{2}m}{2}} &  \geq\mathbb{P}\left(  \left\vert \left\Vert
Ax_{\sigma}\right\Vert ^{2}-\left\Vert x_{\sigma}\right\Vert ^{2}\right\vert
>\epsilon\left\Vert x_{\sigma}\right\Vert ^{2}\right)  \label{eq:Mconc2}\\
&  \geq\mathbb{P}\left(  \left\Vert x\right\Vert ^{2}>\left(  \epsilon
+1\right)  m\sigma^{2}\right)  \nonumber\\
&  +\mathbb{P}\left(  \left\Vert x\right\Vert ^{2}>\left(  1-\epsilon\right)
m\sigma^{2}\right)  \nonumber
\end{align}
and, thus,
\begin{equation}
\mathbb{P}\left(  \left\Vert x\right\Vert ^{2}>\epsilon\right)  \leq
2e^{-\left(  \frac{\epsilon}{m\sigma^{2}}-1\right)  ^{2}\frac{m}{2}%
.}\label{eq:Mconc3}%
\end{equation}
provided that $\epsilon/m\sigma^{2}>1$.

We will also need a concentration inequality for the measurement map $A$, that
was argued to be a uniformly at random subsampled Fourier matrix.
Unfortunately, random Fourier matrices to not directly fulfil a concentration
inequality like \eqref{eq:Mconc1} but only if we restrict $x$ to be sparse.
Suppose $|\operatorname{supp}\left(  x\right)  |\leq k_{s}k_{u}$, then
\cite[Lemma 12.25]{Rahut2013} shows%
\begin{align}
&  \mathbb{P}\left(  \left\vert \left\Vert Ax\right\Vert ^{2}-\left\Vert
x\right\Vert ^{2}\right\vert >\left(  \sqrt{\frac{k_{s}k_{u}}{m}}%
+\epsilon\right)  \left\Vert x\right\Vert ^{2}\right) \nonumber\\
&  \leq2e^{-\frac{\epsilon^{2}m}{2k_{s}k_{u}}\frac{1}{1+2\sqrt{k_{s}k_{u}%
/m}+\epsilon/3}}\nonumber\\
&  \leq2e^{-\frac{\epsilon^{2}m}{2\left(  1+\epsilon\right)  k_{s}k_{u}}}
\label{eq:Mconc4}%
\end{align}
for $m\gg k_{s}k_{u}$.

The assertion of Lemma~\ref{lem:FFTconc} should hold for noisy signals of the
form of \eqref{eq:2edModel}. The model assumes that the signal is of the form
$h^{z}\coloneqq h+z^{\prime}
%{}
\in\mathbb{C}^{n}$, where the support of $h$ is a hierarchically $\left(
k_{u},k_{s}\right)  $-sparse set $\mathcal{A}$ drawn uniformly at random. The
values of the entries are randomly chosen such that $(h^{z})_{\mathcal{A}%
}\sim\mathcal{CN}(0,(\sigma_{h}^{2}+m\sigma^{2}/n^{2})I_{k_{u}k_{s}})$ and
$(h^{z})_{\mathcal{A}^{C}}\sim\mathcal{CN}(0,m\sigma^{2}/n^{2}I_{n-k_{u}k_{s}%
})$.

It is important to note that in this model the signal $h^{z}$ is not sparse
due to the noise contribution. The following proposition will extend the
concentration inequality for subsampled Fourier measurements to this signal model.

\begin{proposition}
\label{prop:FFTconc} Let $A\in\mathbb{C}^{m\times n}$ be a randomly subsampled
FFT and let $h,z^{\prime}$ obey the random model above. Then, it holds:
\begin{align*}
&  \mathbb{E}_{h}\mathbb{P}\left(  \left\vert \left\Vert A\left(  h+z^{\prime
}\right)  \right\Vert ^{2}-\left\Vert h+z^{\prime}\right\Vert ^{2}\right\vert
>\epsilon\left\Vert h+z^{\prime}\right\Vert ^{2}|h\right) \\
&  \leq4e^{-\frac{\left(  \epsilon-3t\right)  ^{2}m}{2\left(  1+\epsilon
-3t\right)  k_{s}k_{u}}}+\frac{6\Gamma\left(  k_{u}k_{s}/2\right)  }%
{\Gamma\left(  k_{u}k_{s}\right)  }\left(  \frac{t^{2}\operatorname{SNR}%
n}{2\left(  1+\epsilon\right)  }\right)  ^{-k_{s}k_{u}}%
\end{align*}
for sufficiently large SNR.
\end{proposition}

\begin{proof}
In the proof we make use of the notation introduced as model
(\ref{eq:2edModel}). To this end, recall that $z=Az^{\prime}\in\mathbb{C}^{m}$
with $z\sim\mathcal{CN}\left(  0,\frac{\sigma^{2}}{n}I_{m}\right)  ,z^{\prime
}\sim\mathcal{CN}\left(  0,\frac{\sigma^{2}m}{n^{2}}I_{n}\right)  $, and,
hence, $\mathbb{E}\left\Vert z\right\Vert ^{2}=\mathbb{E}\left\Vert z^{\prime
}\right\Vert ^{2}$. For the ease of notation, we set $\sigma_{m,n}%
^{2}\coloneqq\sigma^{2}m/n^{2}$ and $\sigma_{n}^{2}\coloneqq\sigma^{2}/n$. It
holds that $\sigma_{n}^{2}m=\sigma_{m,n}^{2}n$.

Expanding $\left\Vert A\left(  h+z^{\prime}\right)  \right\Vert ^{2}%
-\left\Vert h+z^{\prime}\right\Vert ^{2}$ and adding zero yields%
\begin{align*}
&  \left\Vert A\left(  h+z^{\prime}\right)  \right\Vert ^{2}-\left\Vert
h+z^{\prime}\right\Vert ^{2}\\
&  =\left\Vert Ah\right\Vert ^{2}+\left\Vert z\right\Vert ^{2}-\mathbb{E}%
\left\Vert z\right\Vert ^{2}+2\operatorname{Re}\left\langle \left(  Ah\right)
,z\right\rangle \\
&  \quad\quad-\left\Vert h\right\Vert ^{2}-\left\Vert z^{\prime}\right\Vert
^{2}+\mathbb{E}\left\Vert z^{\prime}\right\Vert ^{2}-2\operatorname{Re}%
\left\langle h,z^{\prime}\right\rangle .
\end{align*}

Using the concentration of the norm of $z$ and $z^{\prime}$ as well as the
measurement map, we can bound the individual terms of this expansion. Thus,
for the first step we consider a constant $h$ and only view $z$, $z^{\prime}$
and $A$ as random variables. Let $s,t,o>0$ be small but fixed. Consider the
event $\mathcal{C}_{1}$ defined by:%
\[
\left\{  2\left\vert \operatorname{Re}\left\langle h,z^{\prime}\right\rangle
\right\vert \leq s\sigma_{m,n}\left\Vert h\right\Vert \right\}
\]
which, since $2\operatorname{Re}\left\langle \left(  h\right)  ,z^{\prime
}\right\rangle $ is Gaussian, occurs with probability (at least )
$1-e^{-\frac{s^{2}}{2}}$ \cite[Proposition 7.5, eq. (7.8)]{Rahut2013}.
Similarly, the event $\mathcal{C}_{2}$ defined by:%
\begin{align}
&  \left\{  2\left\vert \operatorname{Re}\left\langle Ah,z\right\rangle
\right\vert \leq s\sigma_{n}\left\Vert Ah\right\Vert \right\}  \label{eq:reAh}%
\\
&  \subseteq\left\{  2\left\vert \operatorname{Re}\left\langle
Ah,z\right\rangle \right\vert \leq s\sigma_{n}\left(  1+2\epsilon\right)
\left\Vert h\right\Vert \right\}  \nonumber
\end{align}
occurs at least with probability $1-e^{-\frac{s^{2}}{2}}+2e^{-\frac
{\epsilon^{2}m}{2\left(  1+\epsilon\right)  k_{s}k_{u}}}$ for sufficiently
large $m$ by eq.~\eqref{eq:Mconc4}. Moreover, the events $\mathcal{C}_{3}$:%
\[
\left\vert \left\Vert z\right\Vert ^{2}-\mathbb{E}\left\Vert z\right\Vert
^{2}\right\vert \leq t\sigma_{n}^{2}m,
\]
and $\mathcal{C}_{4}$:%
\[
\left\vert \left\Vert z^{\prime}\right\Vert ^{2}-\mathbb{E}\left\Vert
z^{\prime}\right\Vert ^{2}\right\vert \leq t\sigma_{m,n}^{2}n
\]
each occur with probabilities $1-2e^{-\frac{t^{2}m}{2}}$ and $1-2e^{-\frac
{t^{2}n}{2}}$ by eq.~\eqref{eq:Mconc2}, respectively. Eventually the event
$\mathcal{C}_{5}$:
\[
\left\Vert z^{\prime}\right\Vert ^{2}<\left(  1+o\right)  n\sigma_{m,n}^{2}%
\]
occurs with probabilty of at least $1-2e^{-\frac{o^{2}n}{2}}$ by
(\ref{eq:Mconc3}). Collect the joint event in $%
%TCIMACRO{\tbigcap _{i}}%
%BeginExpansion
{\textstyle\bigcap_{i}}
%EndExpansion
\mathcal{C}_{i}$. Define $\Lambda\left(  h\right)  \coloneqq|\left\Vert
Ah\right\Vert ^{2}-\left\Vert h\right\Vert ^{2}|$. Conditioning on $h$, we get
with $s_{\epsilon}\coloneqq s\left(  1+2\epsilon\right)  $ and $o_{1}%
\coloneqq1+o$
\begin{align*}
&  \mathbb{P}\left(  \left\{  \left\vert \left\Vert A\left(  h+z^{\prime
}\right)  \right\Vert ^{2}-\left\Vert h+z^{\prime}\right\Vert ^{2}\right\vert
>\epsilon\left\Vert h+z^{\prime}\right\Vert ^{2}\right\}  \cap\mathcal{C}%
%TCIMACRO{\tbigcap _{i}}%
%BeginExpansion
{\textstyle\bigcap_{i}}
%EndExpansion
|h\right)  \\
&  \leq\mathbb{P}\left(  \left\{  \Lambda\left(  h\right)  >\epsilon\left\Vert
h+z^{\prime}\right\Vert ^{2}-2s_{\epsilon}\sigma_{n}\left\Vert h\right\Vert
-2t\sigma_{n}^{2}m\right\}  \cap\mathcal{C}%
%TCIMACRO{\tbigcap _{i}}%
%BeginExpansion
{\textstyle\bigcap_{i}}
%EndExpansion
|h\right)  \\
&  \leq\mathbb{P}\left(  \left\{  \Lambda\left(  h\right)  >\epsilon\left\Vert
h\right\Vert ^{2}-\left\Vert z^{\prime}\right\Vert ^{2}-2s_{\epsilon}%
\sigma_{n}\left\Vert h\right\Vert -2t\sigma_{n}^{2}m\right\}  \cap\mathcal{C}%
%TCIMACRO{\tbigcap _{i}}%
%BeginExpansion
{\textstyle\bigcap_{i}}
%EndExpansion
|h\right)  \\
&  \leq\mathbb{P}\left(  \left\{  \frac{\Lambda\left(  h\right)  }{\left\Vert
h\right\Vert ^{2}}>\epsilon-\frac{2s_{\epsilon}\sigma_{n}}{\left\Vert
h\right\Vert }-\frac{2t\sigma_{n}^{2}m}{\left\Vert h\right\Vert ^{2}}%
-\frac{o_{1}n\sigma_{m,n}^{2}}{\left\Vert h\right\Vert ^{2}}\right\}
\cap\mathcal{C}%
%TCIMACRO{\tbigcap _{i}}%
%BeginExpansion
{\textstyle\bigcap_{i}}
%EndExpansion
|h\right)  .
\end{align*}
Hence, we arrive at:
\begin{align}
&  \mathbb{P}\left(  \left\vert \left\Vert A\left(  h+z^{\prime}\right)
\right\Vert ^{2}-\left\Vert h+z^{\prime}\right\Vert ^{2}\right\vert
>\epsilon\left\Vert h+z^{\prime}\right\Vert ^{2}|h\right)  \nonumber\\
&  \leq4e^{-\frac{\left(  \epsilon-s^{\prime}+t^{\prime}+o^{\prime}\right)
^{2}m}{2\left(  1+\epsilon-s^{\prime}+t^{\prime}+o^{\prime}\right)  k_{s}%
k_{u}}}\label{eq:P-1st}\\
&  +2e^{-\frac{s^{\prime2}\left\Vert h\right\Vert ^{2}}{2\left(
1+2\epsilon\right)  ^{2}\sigma_{n}^{2}}}+2e^{-\left(  \frac{t^{\prime
}\left\Vert h\right\Vert ^{2}}{2m\sigma_{n}^{2}}\right)  ^{2}\frac{m}{2}%
}+2e^{-\left(  \frac{o^{\prime}\left\Vert h\right\Vert ^{2}}{n\sigma_{m,n}%
^{2}}-1\right)  ^{2}\frac{n}{2}}\label{eq:P-2ed}%
\end{align}
where the term in eq. (\ref{eq:P-1st}) is again due to (\ref{eq:Mconc4}) and
$s^{\prime}:=\frac{2s_{\epsilon}\sigma_{n}}{\left\Vert h\right\Vert
},t^{\prime}:=\frac{2t\sigma_{n}^{2}m}{\left\Vert h\right\Vert ^{2}}%
,o^{\prime}:=\frac{o_{1}n\sigma_{m,n}^{2}}{\left\Vert h\right\Vert ^{2}}$ are
the respective substitute variables. Note that we `shuffled' the $\epsilon
$-dependent exponential term from eq. (\ref{eq:reAh}) above into this first
term as well. The terms in eq. (\ref{eq:P-2ed}) are due to bounding the
probabilities of events $\mathcal{C}_{1}^{C}-\mathcal{C}_{5}^{C}$.

We shall now calculate the expectation value with respect to $h$. Since the
expressions are exponentially decaying, it suffices to consider small
$\left\Vert h\right\Vert ^{2}$ with corresponding approximate pdf of
eq.~\eqref{eq:Pdf}. This yields for the $s^{\prime}$-dependent term:
\begin{align*}
&  \frac{1}{\Gamma\left(  k_{u}k_{s}\right)  }\int_{0}^{\infty}e^{-\frac
{s^{\prime2}x}{2\left(  1+\epsilon\right)  \sigma_{n}^{2}}}x^{k_{s}k_{u}%
-1}dx\\
&  =\left(  \frac{s^{\prime2}\operatorname{SNR}n}{2\left(  1+\epsilon\right)
}\right)  ^{-k_{s}k_{u}}%
\end{align*}
using eq.~(\ref{eq:integral}). For the $t^{\prime}$-dependent and the
$o^{\prime}$-dependent terms we get for some (even) $k_{u}$:%
\begin{align*}
&  \frac{1}{\Gamma\left(  k_{u}k_{s}\right)  }\int_{0}^{\infty}e^{-\frac
{t^{\prime2}x^{2}}{2m\sigma_{n}^{4}}}x^{k_{s}k_{u}-1}dx\\
&  =\frac{1}{2\Gamma\left(  k_{u}k_{s}\right)  }\int_{0}^{\infty}%
e^{-\frac{t^{\prime2}u}{2m\sigma_{n}^{4}}}u^{k_{s}k_{u}/2-1}du\\
&  =\frac{\Gamma\left(  k_{u}k_{s}/2\right)  }{2\Gamma\left(  k_{u}%
k_{s}\right)  }\left(  \frac{t^{\prime2}\operatorname{SNR}^{2}n^{4}}{2m^{3}%
}\right)  ^{-k_{s}k_{u}/2}.
\end{align*}
The latter expression decays much faster than the first integral in the limit
of $n\gg m\gg k_{s}k_{u}$. Altogether, we get for any small $\epsilon
,t,\epsilon>3t$:%
\begin{align}
&  \mathbb{P}\left(  \left\vert \left\Vert A\left(  h+z^{\prime}\right)
\right\Vert ^{2}-\left\Vert h+z^{\prime}\right\Vert ^{2}\right\vert
>\epsilon\left\Vert h+z^{\prime}\right\Vert ^{2}|h\right) \nonumber\\
&  \leq4e^{-\frac{\left(  \epsilon-3t\right)  ^{2}m}{2\left(  1+\epsilon
-3t\right)  k_{s}k_{u}}}+\frac{6\Gamma\left(  k_{u}k_{s}/2\right)  }%
{\Gamma\left(  k_{u}k_{s}\right)  }\left(  \frac{t^{2}\operatorname{SNR}%
n}{2\left(  1+\epsilon\right)  }\right)  ^{-k_{s}k_{u}}%
\end{align}
which holds for sufficiently large SNR.
\end{proof}

We turn back again to the proof of the lemma. We prove the result by using a
linear estimator of the form $\Psi:=A^{H}$ which is essentially the first step
of HiHTP and a subsequent energy detection per block.

\begin{proof}
[Proof of Lemma \ref{lem:FFTconc}] Let $h^{z}=h+z^{\prime}$. We defined
$\mathbb{P}_{\overline{\text{sRIP}}}$ as the probability of events of the
form:%
\[
\left\vert \sum_{j\in\omega}\left\vert \left\langle h^{z},v_{j}\right\rangle
\right\vert ^{2}-\sum_{j\in\omega}\left\vert \left\langle \Psi y,v_{j}%
\right\rangle \right\vert ^{2}\right\vert >\epsilon,
\]
where $\omega\in\Omega_{i}$ for some $i\in\lbrack u]$ is some support set of
cardinality $k_{s}$. Straightforward calculation gives%
\begin{align*}
&  \left\vert \sum_{j\in\omega}\left(  \left\vert \left\langle h^{z}%
,v_{j}\right\rangle \right\vert ^{2}-\left\vert \left\langle \Psi
y,v_{j}\right\rangle \right\vert ^{2}\right)  \right\vert \\
&  =\left\vert \sum_{j\in\omega}\left(  \left\vert \left\langle h^{z}%
,v_{j}\right\rangle \right\vert +\left\vert \left\langle \Psi y,v_{j}%
\right\rangle \right\vert \right)  \left(  \left\vert \left\langle h^{z}%
,v_{j}\right\rangle \right\vert -\left\vert \left\langle \Psi y,v_{j}%
\right\rangle \right\vert \right)  \right\vert \\
&  \leq\max_{j\in\omega}\left\vert \left\vert \left\langle h^{z}%
,v_{j}\right\rangle \right\vert -\left\vert \left\langle \Psi y,v_{j}%
\right\rangle \right\vert \right\vert (\sum_{j\in\omega}\left\vert
\left\langle h^{z},v_{j}\right\rangle \right\vert +\sum_{j\in\omega}\left\vert
\left\langle \Psi y,v_{j}\right\rangle \right\vert )
\end{align*}
For the linear estimator under consideration $\Psi=A^{H}$, it holds that
$\left\langle \Psi y,v_{j}\right\rangle =\left\langle A^{H}Ah^{z}%
,v_{j}\right\rangle =\left\langle Ah^{z},Av_{j}\right\rangle $. The second
term can hence be bounded by%
\begin{align*}
&  \sum_{j\in\omega}\left\vert \left\langle h^{z},v_{j}\right\rangle
\right\vert +\sum_{j\in\omega}\left\vert \left\langle \Psi y,v_{j}%
\right\rangle \right\vert \\
&  \leq\sqrt{|\omega|}(\sum_{j\in\omega}\left\vert h_{j}^{z}\right\vert
^{2})^{1/2}+\sum_{j\in\omega}\left\vert \left\langle Ah^{z},Av_{j}%
\right\rangle \right\vert \\
&  \leq\sqrt{|\omega|}(\sum_{j\in\omega}\left\vert h_{j}^{z}\right\vert
^{2})^{1/2}+\left\Vert Ah^{z}\right\Vert \sum_{j\in\omega}\left\Vert
Av_{j}\right\Vert \\
&  \leq\sqrt{|\omega|}(\sum_{j\in\omega}\left\vert h_{j}^{z}\right\vert
^{2})^{1/2}+\left(  1+2\epsilon\right)  \left\Vert h^{z}\right\Vert
\sqrt{|\omega|}\\
&  \leq2\left(  1+2\epsilon\right)  \sqrt{k_{s}}\left\Vert h^{z}\right\Vert .
\end{align*}
with probability:
\begin{equation}
1-2e^{-\frac{\epsilon^{2}m}{2\left(  1+\epsilon\right)  k_{s}k_{u}}}
\label{eq:e-conc}%
\end{equation}
The complementary event can again be shuffled into the term in eq.
(\ref{eq:P-1st}).

Now, we turn to bound the first term. By using the reverse triangle inequality
and the polarization identity:
\[
\left\langle h^{z},v_{j}\right\rangle =\frac{1}{4}\sum_{k=0}^{3}\imath
^{k}\left\Vert h^{z}+\imath^{k}v_{j}\right\Vert ^{2},
\]
we get:%
\begin{align*}
&  \max_{j\in\omega}\left\vert \left\vert \left\langle h^{z},v_{j}%
\right\rangle \right\vert -\left\vert \left\langle Ah^{z},Av_{j}\right\rangle
\right\vert \right\vert \\
&  \leq\max_{j\in\omega}\left\vert \left\langle h^{z},v_{j}\right\rangle
-\left\langle Ah^{z},Av_{j}\right\rangle \right\vert \\
&  \leq\frac{1}{4}\max_{j\in\omega}\left\vert \sum_{k=0}^{3}\imath^{k}\left(
\left\Vert h^{z}+\imath^{k}v_{j}\right\Vert ^{2}-\left\Vert A\left(
h^{z}+\imath^{k}v_{j}\right)  \right\Vert ^{2}\right)  \right\vert \\
&  \leq\max_{j\in\omega,k\in\lbrack4]}\left\vert \left(  \left\Vert
h^{z}+\imath^{k}v_{j}\right\Vert ^{2}-\left\Vert A\left(  h^{z}+\imath
^{k}v_{j}\right)  \right\Vert ^{2}\right)  \right\vert .
\end{align*}
We define $\xi_{j}^{k}\coloneqq h^{z}+\imath^{k}v_{j}$ and abbreviate the
summation $\sum_{i,\omega,j,k}\coloneqq\sum_{i\in\lbrack u]}\sum_{\omega
\in\Omega_{i}}\sum_{j\in\omega}\sum_{k\in\lbrack4]}$. Collecting all the terms
and `averaging' over $h$ yields
\begin{align*}
&  \mathbb{E}_{h}\Pr\left(  \max_{i\in\lbrack u],\omega\in\Omega_{i}%
}\left\vert \sum_{j\in\omega}\left(  \left\vert \left\langle h^{z}%
,v_{j}\right\rangle \right\vert ^{2}-\left\vert \left\langle \Psi
y,v_{j}\right\rangle \right\vert ^{2}\right)  \right\vert >\epsilon|h\right)
\\
&  \leq\sum_{i,\omega,j,k}\mathbb{E}_{h}\Pr\left(  \left\vert \left\Vert
\xi_{j}^{k}\right\Vert ^{2}-\left\Vert A\xi_{j}^{k}\right\Vert ^{2}\right\vert
>\frac{\epsilon}{2\sqrt{k_{s}}\left\Vert h^{z}\right\Vert }|h\right)  \\
&  =\sum_{i,\omega,j,k}\mathbb{E}_{h}\Pr\left(  \left\vert \left\Vert \xi
_{j}^{k}\right\Vert ^{2}-\left\Vert A\xi_{j}^{k}\right\Vert ^{2}\right\vert
>\frac{\epsilon\left\Vert \xi_{j}^{k}\right\Vert ^{2}}{2\sqrt{k_{s}}\left\Vert
h^{z}\right\Vert \left\Vert \xi_{j}^{k}\right\Vert ^{2}}|h\right)  ,
\end{align*}
where we used the union bound. For appropriately small constants $s,o>0$, we
have
\begin{align*}
&  2\sqrt{k_{s}}\left\Vert h^{z}\right\Vert \left\Vert \xi_{j}^{k}\right\Vert
^{2}\\
&  =O\left(  2\sqrt{k_{s}}\left\Vert h\right\Vert ^{3}\left(  2+\sqrt{\left(
s+1\right)  n}\sigma_{m,n}\right)  ^{3}\right)  \\
&  =O\left(  2\sqrt{k_{s}}\left\Vert h\right\Vert ^{3}\left(  2+\sqrt{\left(
s+1\right)  \frac{m}{n}}\sigma\right)  ^{3}\right)
\end{align*}
with probability of at least $1-2e^{-\frac{s^{2}n}{2}}$ by eq.
(\ref{eq:Mconc3}) and provided $\left\Vert h\right\Vert ^{2}\geq1$; moreover,
$\left(  1-o\right)  k_{u}k_{s}\leq\left\Vert h\right\Vert ^{2}\leq\left(
1+o\right)  k_{u}k_{s}$ with probability $1-2e^{-\frac{o^{2}k_{u}k_{s}}{2}}$
by eq. (\ref{eq:Mconc2}). Hence, altogether, we arrive at
\begin{align*}
&  2\sqrt{k_{s}}\left\Vert h^{z}\right\Vert \left\Vert \xi_{j}^{k}\right\Vert
^{2}\\
&  =O\left(  2\sqrt{k_{s}}\left(  \left(  1+o\right)  k_{u}k_{s}\right)
^{\frac{3}{2}}\left(  2+\sqrt{\left(  s+1\right)  \frac{m}{n}}\sigma\right)
^{3}\right)  \\
&  =O\left(  k_{u}^{3/2}k_{s}^{2}\left(  1+o\right)  ^{3/2}\left(  1+s\right)
^{3/2}\right)  .
\end{align*}
{with probability }$1-2e^{-\frac{s^{2}n}{2}}-2e^{-\frac{o^{2}k_{u}k_{s}}{2}}$.

Except for the noise floor $z^{\prime}$ in $h^{z}$, the vectors $\xi_{j}%
^{k}=h^{z}+\imath^{k}v_{j}$ are $\left(  k_{u}k_{s}+1\right)  $-sparse. Fix
some small $\epsilon>0,o=s=1$, and let $\epsilon^{\prime}\in O(\epsilon
k_{u}^{-3/2}k_{s}^{-2})$. We can invoke Prop.~\ref{prop:FFTconc} to give:
\begin{align*}
&  \sum_{i,\omega,j,k}\mathbb{E}_{h}\Pr\left(  \left\vert \left\Vert \xi
_{j}^{k}\right\Vert ^{2}-\left\Vert A\xi_{j}^{k}\right\Vert ^{2}\right\vert
>\epsilon^{\prime}\left\Vert \xi_{j}^{k}\right\Vert ^{2}|h\right)  \\
&  \leq32u\binom{s}{k_{s}}k_{s}\left[  6e^{-\frac{\left(  \epsilon^{\prime
}-3t\right)  ^{2}m}{2\left(  1+\epsilon^{\prime}-3t\right)  \left(  k_{u}%
k_{s}+1\right)  }}\right.  \\
&  \left.  +\frac{6\Gamma\left(  k_{u}k_{s}/2\right)  }{\Gamma\left(
k_{u}k_{s}\right)  }\left(  \frac{t^{2}\operatorname{SNR}n}{2\left(
1+\epsilon^{\prime}\right)  }\right)  ^{-k_{u}k_{s}}\right]  +4e^{-\frac
{k_{u}k_{s}}{2}}%
\end{align*}
for sufficiently large SNR.

Here, we included the term in (\ref{eq:e-conc}) into the first exponential
term in the second line. Assume $1\leq\left(  1+\epsilon^{\prime}\right)
,\left(  1+\epsilon^{\prime}-3t\right)  <2$ with $\epsilon^{\prime}>3t$.
Setting $t=k_{u}^{-2}k_{s}^{-2.5}$ yields $3t(\epsilon^{\prime})^{-1}%
=3t\epsilon^{-1}k_{u}^{3/2}k_{s}^{2}\rightarrow0$ as required and hence:%
\begin{align*}
&  \sum_{i,\omega,j,k}\mathbb{E}_{h}\Pr\left(  \left\vert \left\Vert \xi
_{j}^{k}\right\Vert ^{2}-\left\Vert A\xi_{j}^{k}\right\Vert ^{2}\right\vert
>\epsilon^{\prime}\left\Vert \xi_{j}^{k}\right\Vert ^{2}|h\right)  \\
&  \leq32u\binom{s}{k_{s}}k_{s}e^{-\frac{\epsilon^{\prime2}\left(
1-3t/\epsilon^{\prime}\right)  ^{2}m}{4k_{u}k_{s}}}\\
&  +\frac{32u\binom{s}{k_{s}}k_{s}\Gamma\left(  k_{u}k_{s}/2\right)  }%
{\Gamma\left(  k_{u}k_{s}\right)  }\left(  \frac{t^{2}\operatorname{SNR}n}%
{4}\right)  ^{-k_{u}k_{s}}+4e^{\frac{-k_{u}k_{s}}{2}}\\
&  \leq32u\binom{s}{k_{s}}k_{s}e^{-\frac{\epsilon^{2}m}{O\left(  4k_{u}%
^{4}k_{s}^{5}\right)  }}\\
&  +\frac{32u\binom{s}{k_{s}}k_{s}\Gamma\left(  k_{u}k_{s}/2\right)  }%
{\Gamma\left(  k_{u}k_{s}\right)  }\left(  \frac{\operatorname{SNR}n}%
{4k_{u}^{4}k_{s}^{5}}\right)  ^{-k_{u}k_{s}}+4e^{\frac{-k_{u}k_{s}}{2}}%
\end{align*}
for sufficiently large $k_{u}k_{s}$ (and SNR). By the standard inequality
$\binom{s}{k_{s}}\leq\left(  \frac{es}{k_{s}}\right)  ^{k_{s}}$ for the
binominal coefficient, the first term in the last line can be bounded as:
\[
32u\binom{s}{k_{s}}k_{s}e^{-\frac{\epsilon^{2}m}{O\left(  4k_{u}^{4}k_{s}%
^{5}\right)  }}\leq32u\left(  \frac{es}{k_{s}}\right)  ^{k_{s}}k_{s}%
e^{-\frac{\epsilon^{2}m}{O\left(  4k_{u}^{4}k_{s}^{5}\right)  }}%
\]
while by Stirling's approximation for the gamma functions~$\Gamma(n)=(n-1)!$
for some positive integer $n>0$, we have for the second term:%
\begin{align*}
&  \frac{32u\binom{s}{k_{s}}k_{s}\Gamma\left(  k_{u}k_{s}/2\right)  }%
{\Gamma\left(  k_{u}k_{s}\right)  }\left(  \frac{\operatorname{SNR}n}%
{4k_{u}^{4}k_{s}^{5}}\right)  ^{-k_{u}k_{s}}\\
&  \leq\frac{32}{\sqrt{2}}\operatorname{SNR}^{-k_{u}k_{s}}n^{-2k_{s}}u\left(
\frac{es}{k_{s}}\right)  ^{k_{s}}k_{s}\\
&  \cdot\left(  \frac{e}{2}\right)  ^{k_{u}k_{s}/2}\left(  k_{u}k_{s}\right)
^{-k_{u}k_{s}/2}n^{-k_{u}k_{s}}n^{2k_{s}}\left(  4k_{u}^{4}k_{s}^{5}\right)
^{k_{u}k_{s}}\\
&  \leq\frac{32k_{s}}{\sqrt{2}}\operatorname{SNR}^{-k_{u}k_{s}}n^{-2k_{s}%
}\left(  \frac{e}{k_{s}}\right)  ^{k_{s}}\\
&  \cdot\left(  \frac{8e}{k_{u}k_{s}}\right)  ^{k_{u}k_{s}/2}\left(
\frac{k_{u}^{4}k_{s}^{5}}{n^{(1-3/k_{u})}}\right)  ^{k_{u}k_{s}}\\
&  \leq\frac{32k_{s}}{\sqrt{2}}\frac{\operatorname{SNR}^{-k_{u}k_{s}}%
}{n^{2k_{s}}}\left(  \frac{k_{u}^{4}k_{s}^{5}}{n^{(1-3/k_{u})}}\right)
^{k_{u}k_{s}}%
\end{align*}
where the last term holds for $k_{u}\geq8,k_{s}\geq3$ which gives the final result.
\end{proof}

\subsection{Proof of Theorem \ref{thm:md2}}

The idea of the proof is to estimate $h^{z}=h+z^{\prime}$ directly instead of
$h$. Thus, we consider some $h^{z}$ be fixed. For simplicity, let $\hbar$ be
again the estimator for $h^{z}$. Suppose user $i$ is active (and at most
$k_{u}-1$ arbitrary other users). Denote again the true support of $h$ by
$\mathcal{A}$. As discussed before, the estimation is equivalent:%
\[
y_{\mathcal{B}}=A\left(  h+z^{\prime}\right)
\]
where $z^{\prime}\sim\mathcal{CN}\left(  0,\varkappa^{2}I_{m}\right)  $ where
$\varkappa^{2}=\frac{\sigma^{2}m}{nus}$.

Analogously, to the beginning of the proof of Theorem~\ref{thm:md1}, we have%
\begin{align*}
&  \left\{  \text{missed detection of user }i\right\}  \\
&  \subset\left\{  \Vert(\hbar_{i})_{\omega(\hbar_{i})}\Vert\leq\sqrt{\xi
}\right\}  \cup\left\{  \max_{j\in\overline{\mathcal{A}^{B}}}\Vert(\hbar
_{j})_{\omega(\hbar_{j})}\Vert\geq\Vert(\hbar_{i})_{\omega(\hbar_{i})}%
\Vert\right\}
\end{align*}
where $\omega(\hbar_{i})$, $\omega(\hbar_{j})$ are the support sets of the
$k_{s}$ largest elements of $|\hbar_{i}|$, $|\hbar_{j}|$, respectively. By
Lemma \ref{lem:FFTconc}, HiHTP/HiIHT recovers $\Vert(h_{i}^{z})_{\omega}%
\Vert^{2}$ within $\epsilon$-vicinity with probability $\mathbb{P}%
_{\overline{\text{sRIP}}}\left(  \epsilon\right)  $ uniformly for all $i$ and
all support sets $\omega$. Hence, we have for the first event:%
\begin{align*}
&  \left\{  \Vert(\hbar_{i})_{\omega(\hbar_{i})}\Vert^{2}\leq\xi\right\}  \\
&  \subseteq\left\{  \Vert(\hbar_{i})_{\mathcal{A}_{i}}\Vert^{2}\leq
\xi\right\}  \\
&  \subseteq\left\{  \Vert(h_{i}+z_{i}^{\prime})_{\mathcal{A}_{i}}\Vert
^{2}\leq\xi+\epsilon\right\}  .
\end{align*}
Moreover, for the second event:
\begin{align*}
&  \left\{  \max\nolimits_{j\in\overline{\mathcal{A}^{B}}}\Vert(\hbar
_{j})_{\omega(\hbar_{j})}\Vert^{2}\geq\Vert(\hbar_{i})_{\omega(\hbar_{i}%
)}\Vert\right\}  \\
&  \subseteq\left\{  \max\nolimits_{j\in\overline{\mathcal{A}^{B}}}\Vert
(z_{j}^{\prime})_{\omega(z_{j}^{\prime})}\Vert^{2}+\epsilon\geq\Vert
(h_{i}+z_{i}^{\prime})_{\mathcal{A}_{i}}\Vert^{2}-\epsilon\right\}  \\
&  \subseteq%
%TCIMACRO{\tbigcup \limits_{j\in\overline{\mathcal{A}^{B}}}}%
%BeginExpansion
{\textstyle\bigcup\limits_{j\in\overline{\mathcal{A}^{B}}}}
%EndExpansion
\left\{  \Vert(z_{j}^{\prime})_{\omega(z_{j}^{\prime})}\Vert^{2}+\epsilon
\geq\Vert(h_{i}+z_{i}^{\prime})_{\mathcal{A}_{i}}\Vert^{2}-\epsilon\right\}
\\
&  \subseteq%
%TCIMACRO{\tbigcup \limits_{j\in\overline{\mathcal{A}^{B}}}}%
%BeginExpansion
{\textstyle\bigcup\limits_{j\in\overline{\mathcal{A}^{B}}}}
%EndExpansion
\left\{  k_{s}\max_{k}|(z_{j}^{\prime})_{k}|^{2}+\epsilon\geq\Vert(h_{i}%
+z_{i}^{\prime})_{\mathcal{A}_{i}}\Vert^{2}-\epsilon\right\}  \\
&  \subseteq%
%TCIMACRO{\tbigcup \limits_{j\in\overline{\mathcal{A}^{B}}}}%
%BeginExpansion
{\textstyle\bigcup\limits_{j\in\overline{\mathcal{A}^{B}}}}
%EndExpansion%
%TCIMACRO{\tbigcup \limits_{k\in\mathcal{A}_{j}}}%
%BeginExpansion
{\textstyle\bigcup\limits_{k\in\mathcal{A}_{j}}}
%EndExpansion
\left\{  |\sqrt{k_{s}}(z_{j}^{\prime})_{k}|^{2}+\epsilon\geq\Vert(h_{i}%
+z_{i}^{\prime})_{\mathcal{A}_{i}}\Vert^{2}-\epsilon\right\}  \\
&  \subseteq%
%TCIMACRO{\tbigcup \limits_{j\in\overline{\mathcal{A}^{B}}}}%
%BeginExpansion
{\textstyle\bigcup\limits_{j\in\overline{\mathcal{A}^{B}}}}
%EndExpansion%
%TCIMACRO{\tbigcup \limits_{k\in\mathcal{A}_{j}}}%
%BeginExpansion
{\textstyle\bigcup\limits_{k\in\mathcal{A}_{j}}}
%EndExpansion
\left\{  \left\{  |\sqrt{2k_{s}}(z_{j}^{\prime})_{k}|^{2}\geq\Vert(h_{i}%
+z_{i}^{\prime})_{\mathcal{A}_{i}}\Vert^{2}\right\}  \right.  \\
&  \left.  \cup\left\{  4\epsilon\geq\Vert(h_{i}+z_{i}^{\prime})_{\mathcal{A}%
_{i}}\Vert^{2}\right\}  \right\}  .
\end{align*}
The first term in the last line can be bounded by $n^{-k_{s}}\left(  \frac
{us}{2k_{s}m}\right)  ^{-k_{s}}\operatorname{SNR}^{-k_{s}}s\left(
u-k_{u}\right)  B_{1}\left(  1,k_{s}\right)  $ with using the result for the
block correlator and since the noise in $z_{i}^{\prime}$ is damped with
$\left(  \frac{2m}{us}\right)  $. Now, averaging over all $h$ gives the
result. \qed

\bibliographystyle{IEEEtran}
\bibliography{ccra}

\end{document}